\documentclass{article}

\usepackage{arxiv}

\usepackage[utf8]{inputenc} % allow utf-8 input
\usepackage[T1]{fontenc}    % use 8-bit T1 fonts
\usepackage[unicode=true,bookmarks=false,pdfstartview={FitH},colorlinks=true,citecolor=blue,linkcolor=blue,urlcolor=blue]{hyperref}
\usepackage{url}            % simple URL typesetting
\usepackage{booktabs}       % professional-quality tables
\usepackage{amsfonts}       % blackboard math symbols
\usepackage{nicefrac}       % compact symbols for 1/2, etc.
\usepackage{microtype}      % microtypography
\usepackage{lipsum}		% Can be removed after putting your text content

\usepackage{amsthm}
\usepackage{amsmath}
\usepackage{natbib}
\usepackage{multirow}
\usepackage{url}
\usepackage{amsfonts}
\usepackage{graphicx}
\usepackage{bm}
\usepackage{siunitx}
\usepackage{subcaption}
\usepackage{booktabs, etoolbox}
\usepackage{enumitem}

\DeclareCaptionSubType*{figure}

\usepackage{xspace}
\usepackage{mathtools}
\mathtoolsset{showonlyrefs}
\numberwithin{equation}{section}
\theoremstyle{definition}
\newtheorem*{remark*}{Remark}
\newtheorem{remark}{Remark}
\newcolumntype{L}[1]{>{\raggedright\arraybackslash}p{#1}}
\newcommand{\aD}{(\alpha, D)}
\newcommand{\Xts}{\bm{\mathcal{X}}}
\newcommand{\Yts}{\bm{\mathcal{Y}}}

\newcommand{\tth}{\bm{\theta}}

\newcommand{\X}{\bm{X}}
\newcommand{\Y}{\bm Y}
\newcommand{\Yv}{\Y_{\textnormal{v}}}
\newcommand{\tYv}{\tilde\Y_{\textnormal{v}}}
\newcommand{\Xv}{\X_{\textnormal{v}}}
\newcommand{\Z}{\bm Z}
\newcommand{\F}{\bm F}
\newcommand{\BB}{\bm B}

\newcommand{\dt}{\Delta t}
\newcommand{\dY}{\Delta \bm Y}
\newcommand{\ddY}{\Delta Y}

\newcommand{\dX}{\Delta \bm X}
\newcommand{\ddX}{\Delta X}

\newcommand{\N}{\mathcal{N}}
\newcommand{\MN}{\textnormal{MatNorm}}
\newcommand{\mmu}{\bm \mu}
\newcommand{\bbe}{\bm \beta}
\newcommand{\SSi}{\bm \Sigma}

\newcommand{\pph}{{\bm \varphi}}
\newcommand{\rrh}{\bm \rho}
\newcommand{\VV}{\bm V}
\newcommand{\iid}{\stackrel {\textrm{iid}}{\sim}}
\newcommand{\tmin}{t_{\textnormal{min}}}
\newcommand{\tmax}{t_{\textnormal{max}}}
\newcommand{\rv}[3][1]{#2_{#1},\ldots,#2_{#3}}
\newcommand{\water}{\textnormal{H2O}}
\newcommand{\gly}{\textnormal{GLY}}
\newcommand{\kB}{k_B}
\newcommand{\wt}{wt\%}
\newcommand{\om}{\omega}

\newcommand{\gga}{\bm{\gamma}}
\newcommand{\bt}{\bm{\Upsilon}}
\newcommand{\It}{I_{\bt}}
\newcommand{\bz}{\bm{0}}
\let\vec\relax
\DeclareMathOperator{\vec}{vec}
\newcommand{\eps}{\bm \varepsilon}
\newcommand{\widebar}[1]{\overline{#1}}

\DeclareMathOperator{\MSD}{\textsc{msd}}
\DeclareMathOperator{\acf}{\textsc{acf}}
\newcommand*\ud{\mathop{}\!\mathrm{d}}
\DeclareMathOperator{\cov}{cov}
\DeclareMathOperator{\var}{var}

\DeclareMathOperator{\tr}{tr}

\DeclareMathOperator*{\argmax}{arg\,max}
\DeclareMathOperator*{\argmin}{arg\,min}

\newcommand{\elp}{\ell_{\textnormal{prof}}}
\newcommand{\elc}{\ell_{\textnormal{C}}}

\newcommand{\bigO}{\mathcal O}
\newcommand{\fbm}{B^\alpha}
\DeclareMathOperator{\arma}{ARMA}
\newcommand{\armapq}{\arma(p,q)}
\DeclareMathOperator{\ma}{MA}
\DeclareMathOperator{\ar}{AR}
\newcommand{\mao}{\ma(1)}
\newtheorem{theorem}{Theorem}
\newtheorem{lemma}{Lemma}
\DeclareMathOperator{\se}{se}
\newcommand{\rcov}{\textnormal{P}_{95}}
\newcommand{\SNR}{\textnormal{SNR}}
\newcommand{\id}{\mathfrak 1}
\newcommand{\nd}{k}
\newcommand{\np}{d}
\newcommand{\aeff}{\alpha_{\textnormal{eff}}}
\newcommand{\Deff}{D_{\textnormal{eff}}}
\newcommand{\PEO}{\textnormal{PEO}}
\newcommand{\HBE}{\textnormal{HBE}}

\newcommand{\Zs}{\Z^\star}
\newcommand{\fs}{f^\star}
\newcommand{\eet}{\bm \eta}
\renewcommand{\sp}[1]{^{(#1)}}
\newcommand{\dsc}{\mathcal S}
\let\vec\relax
\DeclareMathOperator{\vec}{vec}
\DeclareMathOperator{\Toep}{Toep}
\newcommand{\FFT}{\bm{\mathcal F}}
\newcommand{\vv}{\bm v}
\newcommand{\email}[1]{\href{mailto:#1}{\texttt{#1}}}

\title{Measurement Error Correction in Particle Tracking Microrheology}

\author{
  Yun Ling \\
  Department of Statistics and Actuarial Science \\
  University of Waterloo \\
  \And
  Martin Lysy\thanks{Corresponding author: \email{mlysy@uwaterloo.ca}.} \\
  Department of Statistics and Actuarial Science \\
  University of Waterloo \\
  \And
  Ian Seim \\
  Department of Applied Physical Sciences \\
  University of Carolina Chapel Hill \\
  \And
  Jay M. Newby \\
  Department of Mathematical and Statistical Sciences \\
  University of Alberta \\
  \And
  David B. Hill \\
  Marsico Lung Institute \\
  Department of Physics and Astronomy \\
  University of Carolina Chapel Hill \\
  \And
  Jeremy Cribb \\
  Department of Physics and Astronomy \\
  University of Carolina Chapel Hill \\
  \And
  M. Gregory Forest \\
  Department of Applied Physical Sciences \\
  Department of Biomedical Engineering \\
}

\date{November 14, 2019}

\renewcommand{\headeright}{Ling et al.}
\renewcommand{\undertitle}{}

\begin{document}
\maketitle

\begin{abstract}
  In diverse biological applications, particle tracking of passive microscopic species has become the experimental measurement of choice -- when either the materials are of limited volume, or so soft as to deform uncontrollably when manipulated by traditional instruments.
  In a wide range of
  % biological
  particle tracking experiments, a ubiquitous finding is that the mean squared displacement (MSD) of particle positions exhibits a power-law signature, the parameters of which
  % contain important
  reveal valuable information about the
  % frequency-dependent
  viscous and elastic properties of various biomaterials.
  However, MSD measurements are typically contaminated by complex and interacting sources of instrumental noise.
  As these often affect the high-frequency bandwidth to which MSD estimates are
  particularly
  % especially
  sensitive, inadequate error correction can lead to severe bias in power law estimation and thereby, the inferred viscoelastic properties.
  In this article, we propose a novel strategy to filter high-frequency noise from particle tracking measurements. 
  % measurements of particle positions.
  Our filters are shown theoretically to cover a broad spectrum of high-frequency noises,
  % noise sources,
  and lead to a parametric estimator of MSD power-law coefficients 
  for which an efficient computational implementation is presented.
  Based on numerous analyses of experimental and simulated data, results suggest our methods perform very well compared to other denoising procedures.
\end{abstract}

% keywords can be removed

%-------------------------------------------------------------------------------

\section{Introduction}

With the development of high-resolution microscopy,
% \correct{insert references to the technology perhaps???? or wait for reviews???? or move Mason et al., 1997 to here},
single-particle tracking has emerged as an invaluable tool in the study of biophysical and transport properties of
diverse soft materials~\citep[e.g.,][]{mason.etal97}.
Examples of applications include cellular membrane dynamics~\citep{saxton.jacobson97}, drug delivery mechanisms~\citep{suh.etal05}, properties of colloidal particles~\citep{lee.etal07},
% \correct{Greg: not sure what this means. so I change ``process'' to ``mechanism''}
mechanisms of virus infection~\citep{vander.etal08}, microrheology of complex fluids and living cells~\citep{mason.etal97,wirtz09} and functional analyses of the cytoskeleton~\citep{gal.etal13}.

\emph{Passive} single-particle tracking refers to experiments in which microscale probes and/or pathogens (e.g., viruses) are recorded without external forcing, producing high-resolution time series of particle positions from which dynamical properties of the transport medium are inferred.  In many of these experiments, the resulting analysis hinges pivotally on the measurement of particles' mean square displacement (MSD), which for a $k$-dimensional particle trajectory $\X(t) = \bigl(X_1(t), \ldots, X_\nd(t)\bigl)$ (with $k \in \{1,2,3\}$ depending on the experiment) is given by
\begin{equation}\label{eq:msd}
  \MSD_{\X}(t) = \frac 1 \nd \times E\bigl[\lVert \X(t) - \X(0) \rVert^2\bigr] = \frac 1 \nd \times \sum_{j=1}^\nd E\bigl[\lvert X_j(t) - X_j(0) \rvert^2\bigr].
\end{equation}

For particles diffusing %\correct{in dimension $d$}
in viscous media (e.g., water, glycerol),
% \correct{the observed increment process is}
the position time series are accurately modeled by Brownian motion. The MSD is then linear in time,
\begin{equation}\label{eq:diff}
  \MSD_{\X}(t) = 2D t,
\end{equation}
and the diffusion coefficient $D$ is determined by the Stokes-Einstein relation~\citep{einstein56,edward70}
% , the diffusion coefficient $D$ is determined by
\begin{equation}\label{eq:stokes}
  D = \frac{\kB T}{6 \pi \eta r},
\end{equation}
where $r$ is the particle radius, $T$ is temperature, $\eta$ is the viscosity of the medium, and $k_B$ is the Boltzmann's constant.
However, due to the microstructure of large molecular weight biopolymers (e.g., mucins in mucosal layers), most biological fluids are \emph{viscoelastic}. In such fluids, a nearly ubiquitous experimental finding has been that the MSD has sublinear power-law scaling over a given range of timescales,
\begin{equation}\label{eq:subdiff}
  \MSD_{\X}(t) \sim 2 D t^\alpha, \qquad \tmin < t < \tmax, \quad 0 < \alpha < 1.
\end{equation}
This phenomenon is referred to as %\correct{transient}
\emph{subdiffusion}. Due to its pervasiveness, interpretation of the subdiffusion parameters $\aD$ has far-reaching consequences for numerous biological applications, for example: distinguishing signatures of healthy versus pathological human bronchial epithelial mucus~\citep{hill.etal14}; cytoplasmic crowding~\citep{weiss.etal04}; local viscoelasticity in protein networks~\citep{amblard.etal96}; dynamics of telomeres in the nucleus of mammalian cells~\citep{bronstein.etal09}; and microstructure dynamics of entangled F-Actin networks~\citep{wong.etal04}.

Unlike viscous fluids exhibiting ordinary (linear) diffusion, the precise manner in which the properties of a viscoelastic fluid determine its subdiffusion parameters $\aD$ is unknown, 
% the exact relation between subdiffusion parameters $\aD$ and properties of viscoelastic fluids is unknown,
such that $\aD$ must be estimated from particle-tracking data. To this end, a widely-used approach is to apply ordinary least-squares to
a nonparametric estimate of the MSD~\citep[e.g.,][]{qian.etal91}.
% a non-parametric estimate of the MSD against time on the log-log scale~\citep[e.g.,][]{qian.etal91}.
While minimal modeling assumptions suffice to make this subdiffusion estimator consistent~\citep{michalet10}, for finite-length trajectories, the nonparametric MSD estimator at longer timescales is severely biased~\citep{mellnik.etal16}.  Therefore, in practice a good portion of the MSD must be discarded,
% in practice the information about longer timescales is typically discarded,
at the expense of considerable loss in statistical efficiency.  In contrast, fully parametric subdiffusion estimators specify a complete stochastic process for $\X(t)$ as a function of $\aD$~\citep[e.g.,][]{berglund10, lysy.etal16, mellnik.etal16}, whereby optimal statistical efficiency is achieved via likelihood inference.  However, the accuracy of these parametric estimators critically depends on the adequacy of the parametric model, and particle tracking measurements are well known to be corrupted by various sources of experimental noise.

Noise in single-particle tracking experiments 
can be categorized roughly into two types.  Low-frequency noise, originating primarily from slow drift currents in the fluid itself, is typically removed from particle trajectories by way of various linear detrending methods~\citep[e.g.,][]{fong.etal13, rowlands.so13, koslover.etal16, mellnik.etal16}.  
In contrast, high-frequency noise can be due to a variety of reasons: mechanical vibrations of the instrumental setup;  particle displacement while the camera shutter is open; noisy estimation of true position from the pixelated microscopy image; error-prone tracking of particle positions when they are out of the camera focal plane.
% ~\citep[for a systematic overview of various causes, see][]{deschout.etal14}.
A systematic review of high-frequency or \emph{localization} errors in single-particle tracking is given by~\cite{deschout.etal14}.  
The effect of %high-frequency
such noise is to distort the MSD at the shortest observation timescales.  Since fully-parametric models extract far more information about $\aD$ from short timescales than long ones, their accuracy in the presence of high-frequency noise can suffer considerably.

In a seminal work,~\cite{savin.doyle05} present a theoretical model for localization error, encompassing most of the approaches reviewed by~\cite{deschout.etal14}.  The parameters of the Savin-Doyle model can be derived either from first-principles~\citep[for instance, by analyzing uncertainty in position-extraction algorithms, e.g.,][]{mortensen.etal10,chenouard.etal14,kowalczyk.etal14,burov.etal17}, or empirically~\citep[via signal-free control experiments, e.g.,][]{savin.doyle05,deschout.etal14}.  Model-based methods for estimating localization error have also been proposed, under the assumption of ordinary diffusion $\alpha = 1$~\cite[e.g.,][]{michalet10, berglund10,michalet.berglund12,vestergaard.etal14,ashley.andersson15,calderon16}.

The Savin-Doyle theoretical framework accounts for a wide range of experimental errors.  However, due to the extreme complexity and inter-dependence between various sources of localization error, the Savin-Doyle model cannot account for them all.  This is illustrated in the control experiment of Figure~\ref{fig:controla}, where trajectories of \SI{1}{\micro\meter} diameter tracer particles are recorded in water, for which it is known that $\alpha = 1$ and for which $D$ may be determined  theoretically by the Stokes-Einstein relation~\eqref{eq:stokes}.  However, the Savin-Doyle model estimates both of these parameters with considerable bias (Figure~\ref{fig:controlb}).
\begin{figure}[!htb]
  \centering
  \begin{subfigure}{\textwidth}
    \includegraphics[width = 1\textwidth]{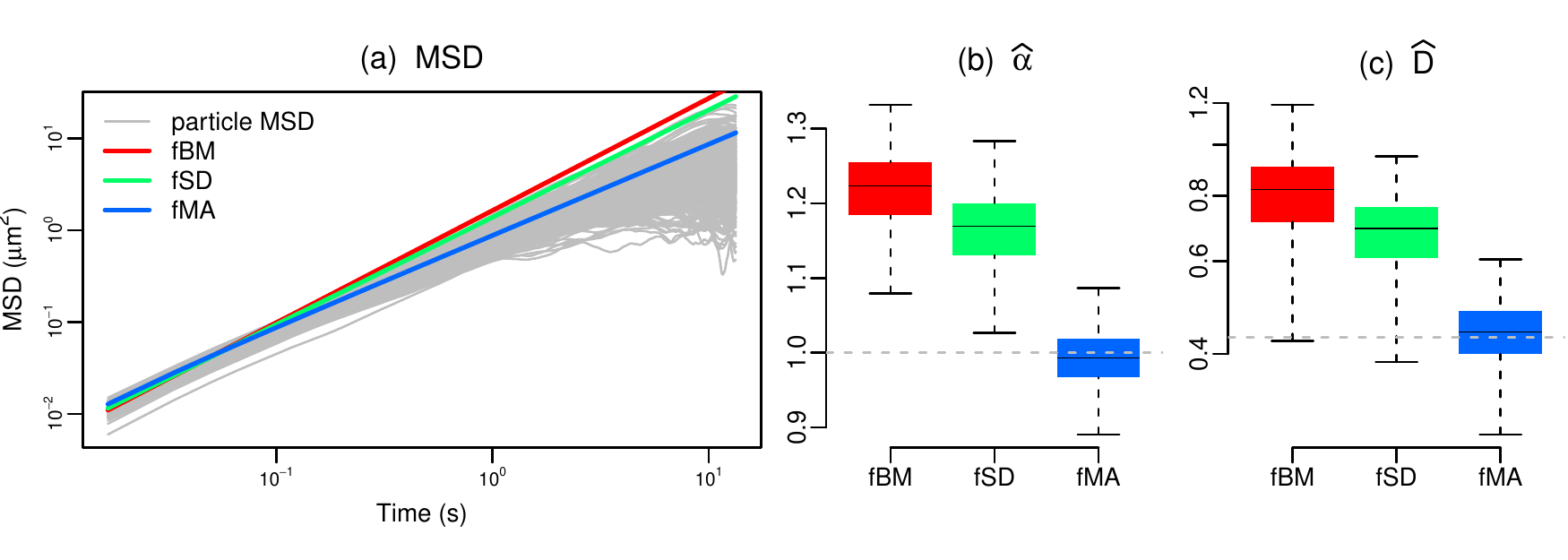}\phantomsubcaption\label{fig:controla}\phantomsubcaption\label{fig:controlb}\phantomsubcaption\label{fig:controlc}
  \end{subfigure}
  \caption{(a) Pathwise empirical MSD for $1931$ particles of diameter \SI{1}{\micro\meter} recorded in water at $\dt = \SI{1/60}{\second}$ for a duration of $\SI{30}{\second}$ ($N = 1800$ observations). Straight lines correspond to fitted MSDs for three parametric models: fractional Brownian Motion (fBM), fBM with Savin-Doyle noise correction (fSD), and fBM with one of the noise correction models proposed in this paper (fMA).
    (b-c) Estimated values of $\alpha$ and $D$ for each particle and parametric model. The predicted values from Stokes-Einstein theory are given by the horizontal dashed lines.}
  \label{fig:control}
\end{figure}

In this article, we propose a likelihood-based filtering method to correct for localization errors, complementing the theoretical Savin-Doyle approach.  Our filters can be readily applied to any parametric model of particle dynamics, and are demonstrated theoretically to cover a very broad spectrum of high-frequency noises.  
We show how to combine our filters with parametric methods of low-frequency drift correction,
and estimate all parameters of both subdiffusion and noise models in a computationally efficient manner. 
Extensive simulations and analyses of experimental data suggest that
% a one-parameter version of our filters
our filters perform remarkably well, both for estimating the true values of $\aD$, and compared to state-of-the-art high-frequency
% error correction
denoising procedures
(e.g., Figure~\ref{fig:controlc}).

The remainder of the article is organized as follows.  In Section~\ref{sec:background} we review a number of existing subdiffusion estimators and high-frequency error-correction techniques.  In Section~\ref{sec:filter} we present our family of high-frequency filters, with theoretical justification for the proposed construction.  Sections~\ref{sec:sim} and~\ref{sec:exper} contain analyses of numerous simulated and real particle-tracking experiments 
% simulation results and analyses of numerous viscous and viscoelastic particle tracking experiments
comparing our proposed subdiffusion estimators to existing alternatives.  Section~\ref{sec:disc} offers concluding remarks and directions for future work. 

%-------------------------------------------------------------------------------

\section{Existing Subdiffusion Estimators}\label{sec:background}

\subsection{Semiparametric Least-Squares Estimator}

Let $\X = (\rv [0] \X N)$, $\X_n = \X(n\cdot \dt)$, denote the discrete-time observations of a given particle trajectory $\X(t)$ recorded at frequency $1/\dt$. 
Assuming that %the position process
$\X(t)$ has second-order stationary increments,
\begin{equation}\label{eq:statincr}
  E\bigl[\lVert \X(s+t) - \X(s) \rVert^2\bigr] = E\bigl[\lVert \X(t) - \X(0) \rVert^2\bigr],
\end{equation}
a standard nonparametric estimator for the particle MSD is given by
\begin{equation}\label{eq:msdemp}
  \widehat \MSD_{\X}(n \cdot \dt) = \frac{1}{\nd\cdot(N-n+1)} \sum_{i=0}^{N-n} \lVert \X_{n+i} - \X_i \rVert^2.
\end{equation}
Based on the linear relation
\begin{equation}\label{eq:logreg}
  \log \MSD_{\X}(t) = \log 2D + \alpha \log t
\end{equation}
over the subdiffusion timescale $t \in (\tmin, \tmax)$, a commonly-used subdiffusion estimator~\citep[e.g.,][]{gal.etal13} is obtained from the  least-squares regression of $y_n = \log \bigl(\widehat \MSD_{\X}(n\cdot \dt)\bigr)$ onto $x_n = \log(n\cdot \dt)$, namely
\begin{equation}\label{eq:ols}
  \hat \alpha = \frac{\sum_{n=0}^N(y_n - \bar y)(x_n - \bar x)}{\sum_{n=0}^N(x_n - \bar x)^2}, \qquad \hat D = \tfrac 1 2 \exp(\bar y - \hat \alpha \bar x).
\end{equation}

The least-squares subdiffusion estimator~\eqref{eq:ols} is easy to implement, and it is consistent under the minimal assumption of~\eqref{eq:statincr} and when the power-law scaling~\eqref{eq:subdiff} holds for all $t > \tmin$~\citep{sikora.etal17}.
% the upper end of the subdiffusion timescale~\eqref{eq:subdiff} is such that \mbox{$\tmax \to \infty$}~\citep{sikora.etal17}.
However, the least-squares estimator also presents two major drawbacks.  First, the errors underlying the regression~\eqref{eq:logreg} are neither homoscedastic nor uncorrelated~\citep{sikora.etal17}, such that~\eqref{eq:ols} is statistically inefficient.  Second, it is common practice to account for low-frequency noise by calculating the empirical MSD~\eqref{eq:msdemp} from the drift-subtracted positions
\begin{equation}\label{eq:dsp}
  \tilde \X_n = (\X_n - \X_0) - n \cdot \widebar{\dX},
  % \qquad \widebar{\dX} = \frac 1 N \sum_{n=1}^N \X_n - \X_{n-1},
\end{equation}
where $\widebar{\dX} = \frac 1 N \sum_{n=1}^N (\X_n - \X_{n-1})$ is the average displacement over the interobservation time $\dt$.  However, a straightforward calculation~\citep{mellnik.etal16} shows that $\tilde \X_N = 0$, such that $\widehat{\MSD}_X(n \cdot \dt)$ becomes increasingly biased towards zero as $n$ approaches $N$.  Consequently, a widely-reported figure~\citep[e.g.,][]{weihs.etal07} suggests that, prior to fitting~\eqref{eq:ols}, the largest 30\% of MSD lag times are discarded, thus severely compounding the inefficiency of the least-squares subdiffusion estimator when low-frequency noise correction is applied.

\subsection{Fully-Parametric Subdiffusion Estimators}\label{sec:parest}

While the semiparametric estimator~\eqref{eq:ols} operates under minimal modeling assumptions, complete specification of the stochastic process $\X(t)$ provides not only a considerable increase in statistical efficiency~\citep[e.g.,][]{mellnik.etal16}, but in fact is necessary to establish dynamical properties of particle-fluid interactions 
% -- such as first-passage times of microparticle pathogens through protective mucosal layers --
which cannot be determined from second-order moments (such as the MSD) alone~\citep{gal.etal13,lysy.etal16}.  A convenient framework for stochastic subdiffusion modeling is the location-scale model of~\cite{lysy.etal16},
\begin{equation}\label{eq:lsmodel}
  \X(t) = \sum_{j=1}^\np \bbe_j f_j(t) + \SSi^{1/2} \Z(t),
\end{equation}
where $f_1(t), \ldots f_\np(t)$ are known functions accounting for low-frequency drift (typically linear, $f_1(t) = t$, and occasionally quadratic, $f_2(t) = t^2$), $\rv \bbe \np \in \mathbb R^\nd$ are regression coefficients, $\SSi_{\nd \times \nd}$ is a variance matrix, and $\Z(t) = \bigl(Z_1(t), \ldots, Z_\nd(t)\bigr)$ are iid continuous stationary-increments (CSI) Gaussian processes with mean zero and MSD parametrized by $\pph$,
\begin{equation}\label{eq:csimsd}
  \MSD_Z(t) = E\bigl[\lVert Z_j(t) - Z_j(0) \rVert^2\bigr] = \eta(t \mid \pph),
\end{equation}
such that the MSD of the drift-subtracted process $\tilde \X(t) = \X(t) - \sum_{j=1}^\np \bbe_j f_j(t)$ is given by
\begin{equation}\label{eq:csimsddrift}
  \MSD_{\tilde \X}(t) = \tfrac{1}{\nd} \tr(\SSi) \cdot \eta(t \mid \pph).
\end{equation}

Perhaps the simplest parametric subdiffusion model sets $Z_j(t) = B_\alpha(t)$ to be fractional Brownian Motion (fBM)~\citep[e.g.,][]{szymanski.weiss09,weiss13}, a mean-zero CSI Gaussian process with covariance function
\begin{equation}\label{eq:fbm-cov}
  \cov\bigl(B_\alpha(t), B_\alpha(s)\bigr) = \tfrac 1 2 (|t|^\alpha + |s|^\alpha - |t-s|^\alpha), \qquad 0 < \alpha < 2.
\end{equation}
Indeed, as the covariance function of a CSI process is completely determined by its MSD, fBM is the only (mean-zero) CSI Gaussian process exhibiting \emph{uniform} subdiffusion,
\begin{equation}\label{eq:fbm-msd}
  \MSD_{B_\alpha}(t) = t^\alpha, \qquad 0 < t < \infty,
  % \qquad 0 < \alpha < 1
\end{equation}
in which case the diffusivity coefficient is given by
\[
  D = \frac{1}{2k} \times \tr(\SSi).
\]
Other examples of driving CSI processes are the confined diffusion model of~\cite{ernst.etal17} and the viscoelastic Generalized Langevin Equation (GLE) of~\cite{mckinley.etal09}, both of which exhibit \emph{transient}
% (anomalous)
subdiffusion, i.e., power-law scaling only on a given timescale $t \in (\tmin, \tmax)$. In this case, the subdiffusion parameters $\aD$ become functions of the other parameters, namely $\alpha = \alpha(\pph)$ and $D = D(\pph, \SSi)$.  We shall revisit these transient subdiffusion models in Section~\ref{sec:sim}.

Parameter estimation for the location-scale model~\eqref{eq:lsmodel} can be done by maximum likelihood.  Let $\dX_n = \X_{n+1} - \X_n$ denote the $n$th trajectory increment, and $\dX = (\rv [0] \dX {N-1})$.  Then $\dX$ are consecutive observations of a stationary Gaussian time series with autocorrelation function
\[
  \acf_{\dX}(h) = \cov(\dX_n, \dX_{n+h}) = \SSi \times \gamma(h \mid \pph),
\]
where
\[
  \gamma(n \mid \pph) = \tfrac 1 2 \times \Bigl\{\eta(|n-1| \cdot \dt \mid \pph) + \eta(|n+1| \cdot \dt \mid \pph) - 2 \eta(|n|\cdot \dt \mid \pph)\Bigr\},
\]
such that the increments follow a matrix-normal distribution (defined in Appendix~\ref{appendix:profile}),
\begin{equation}\label{eq:matnorm}
  \dX_{N\times k} \sim \MN(\F\bbe, \VV_\pph, \SSi),
  % \quad \iff \quad \vec(\dX) \sim \N(\vec(\F\bbe), \SSi \otimes \VV_\pph),
\end{equation}
where $\bbe_{\np \times \nd} = [\bbe_1 \mid \cdots \mid \bbe_\np]'$, $\F_{N \times \np}$ is a matrix with elements $F_{nm} = f_m((n+1)\cdot \dt) - f_m(n\cdot \dt)$, 
and $\VV_\pph$ is an $N \times N$ Toeplitz matrix with element $(n,m)$ given by $V_\pph^{(n,m)} = \gamma(n-m \mid \pph)$, such that the log-likelihood function is given by
\begin{equation}\label{eq:lsloglik}
  \begin{split}
    \ell(\pph, \bbe, \SSi \mid \dX) = & -\frac 1 2 \tr\left\{\SSi^{-1}(\dX - \F \bbe)'\VV_{\pph}^{-1}(\dX - \F\bbe)\right\} \\
    & - \frac N 2 \log |\SSi| - \frac \nd 2 \log |\VV_{\pph}|.
  \end{split}
\end{equation}

In order to calculate the MLE of $\tth = (\pph, \bbe, \SSi)$, model~\eqref{eq:lsmodel} has two appealing properties.  First, for given $\pph$, the conditional MLEs of $\bbe$ and $\SSi$ can be obtained analytically as shown in Appendix~\ref{appendix:profile}, such that the optimization problem can be reduced by $2\nd + {\nd\choose 2}$ dimensions by calculating the profile likelihood $\elp(\pph \mid \dX) = \max_{\bbe, \SSi} \ell(\pph, \bbe, \SSi \mid \dX)$.  Second, we show in Appendix~\ref{appendix:profile} that the computational bottleneck in $\elp(\pph \mid \dX)$ involves the calculation of $\VV_{\pph}^{-1}$ and its log-determinant.  While the computational cost of these operations is $\bigO(N^3)$ for general variance matrices, for Toeplitz matrices it is only $\bigO(N^2)$ using the Durbin-Levinson algorithm~\citep{levinson47,durbin60}, or more recently, only $\bigO(N \log^2N)$ using the Generalized Schur algorithm~\citep{kailath.etal79,ammar.gragg88,ling.lysy17}.

\subsection{Savin-Doyle Noise Model}\label{sec:sdmodel}

In order to characterize high-frequency noise in particle tracking experiments,~\cite{savin.doyle05} decompose it into so-called \emph{static} and \emph{dynamic} sources.  Static noise is due to measurement error in the recording of the position of the particle at a given time.  Thus, if $\X_n$ denotes the true particle position at time $t = n \cdot\dt$, and $\Y_n$ is its recorded value, then~\citeauthor{savin.doyle05} suggest the additive error model
\begin{equation}\label{eq:staterr}
  \Y_n = \X_n + \eps_n,
\end{equation}
where $\eps_n$ is a $\nd$-dimensional stationary process independent of $\X(t)$.  Thus, if the autocorrelation of the static noise is denoted as
\begin{equation}\label{eq:acf}
  \acf_{\eps}(n) = \cov(\eps_m, \eps_{m+n}),
\end{equation}
the MSD of the observations becomes
\begin{equation}
  \begin{aligned}
    \MSD_{\Y}(n) & = \tfrac 1 \nd \times E\bigl[\lVert \Y_n - \Y_0 \rVert^2\bigr] \\
    & = \MSD_{\X}(n) + \tfrac 1 \nd \times 2\cdot \tr\bigl(\acf_{\eps}(0) - \acf_{\eps}(n)\bigr).
  \end{aligned}
\end{equation}
\citeauthor{savin.doyle05} describe how to estimate the temporal dynamics of $\eps_n$ by recording immobilized particles, i.e., for which it is known that $\X_n \equiv 0$.  Over a wide range of signal-to-noise ratios, they report that $\eps_n$ is effectively white noise,
\[
  \acf_{\eps}(n) = \SSi_{\varepsilon} \cdot \id(n = 0),
\]
a result corroborated by many other experiments~\citep[for example, see references in][Figure 2]{deschout.etal14}. For the canonical trajectory model of fractional Brownian motion, $\MSD_{\X}(t) = 2D t^\alpha$, white static noise has the effect of raising the MSD at the shortest timescales, as seen in Figure~\ref{fig:locerrb}.
\begin{figure}[!htb]
  \centering
  \begin{subfigure}{\textwidth}
    \includegraphics[width = 1\textwidth]{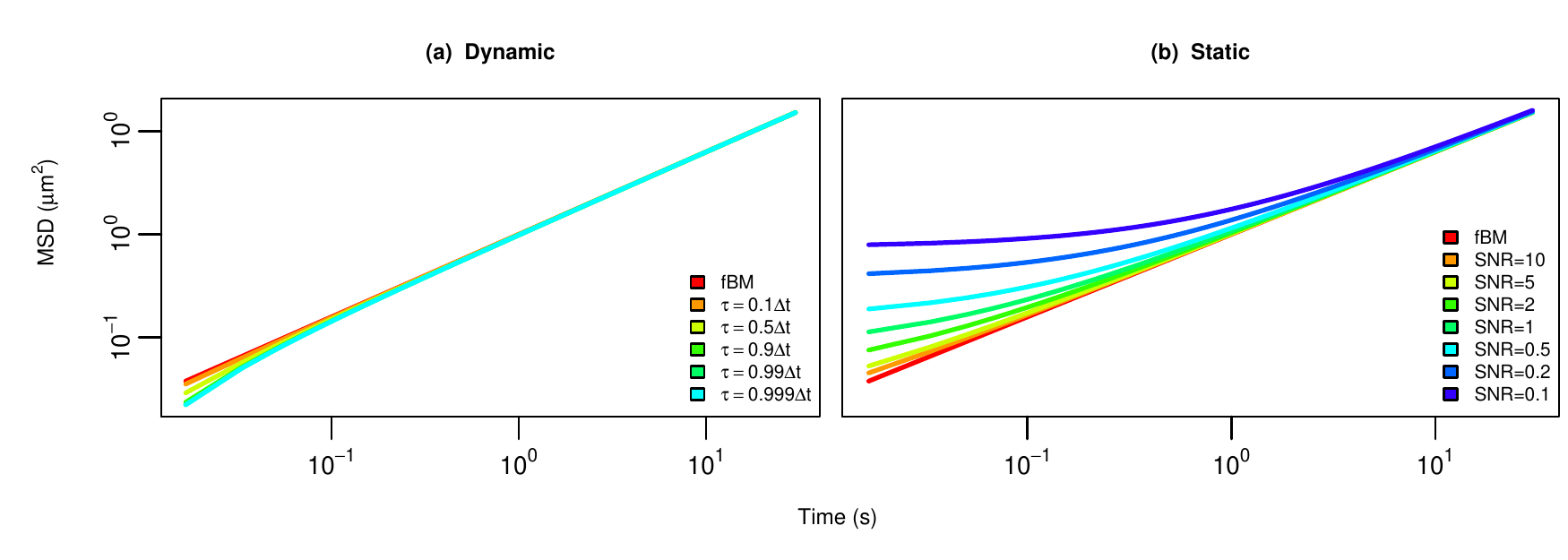}
    \phantomsubcaption\label{fig:locerra}
    \phantomsubcaption\label{fig:locerrb}
  \end{subfigure}
  \caption[MSD of localization errors]{Effect of localization error on the MSD of an fBM process $X(t) = \fbm_t$ with $\alpha = 0.8$ and $\dt = 1/60$.
    (a) Dynamic error, as a function of exposure time $\tau$.
    (b) Static error, as a function of the signal-to-noise ratio, $\SNR = \var(\Delta \fbm_n) / \var(\varepsilon_n)$.}
  \label{fig:locerr}
\end{figure}

In contrast to static noise,~\citeauthor{savin.doyle05} define dynamic noise as originating from movement of the particle during the camera frame exposure time.  Thus, if the camera exposure time is $\tau < \dt$ (as it must be less than the framerate), the recorded position of the particle at time $t = n\cdot\dt$ is
\begin{equation}\label{eq:dynerr}
  \Y_n = \frac 1 \tau \int_0^\tau \X(n\cdot\dt - s) \ud s.
\end{equation}
The dynamic-error MSD for an fBM process $X(t) = \fbm_t$ is given in Appendix~\ref{appendix:fdl-acf}.  Larger values of $\tau$ have the effect of lowering the MSD at the shortest timescales, as seen in Figure~\ref{fig:locerra}.

Combining static and dynamic models, the Savin-Doyle localization error model is
\begin{equation}\label{eq:locmod}
  \Y_n = \frac 1 \tau \int_0^\tau \X(n\cdot\dt - s) \ud s + \eps_n.
\end{equation}
When $\X(t) = \sum_{m=1}^\np \bbe_m f_m(t) + \SSi^{1/2} \Z(t)$ follows the location-scale model~\eqref{eq:lsmodel}, and the static noise has the simplified form $\SSi_{\varepsilon} = \sigma^2 \cdot \SSi$, parametric inference can be conducted using the computationally efficient methods of Section~\ref{sec:parest}.  Explicit calculations for the fBM process with $\MSD_Z(t) = t^\alpha$ are given in Appendix~\ref{appendix:fdl-acf}.

Thus, the fBM + Savin-Doyle (fSD) model has three MSD parameters: $\pph = (\alpha, \tau, \sigma)$. 
%\correct{not D ???  if not, we want to be sure the reader knows where D fits in} 
Its maximum likelihood estimates of the subdiffusion parameters $\aD$ are $\hat \alpha$ and $\hat D = (1/2\nd) \cdot \tr(\hat \SSi)$. While these estimates successfully correct for many types of high-frequency measurement errors, the fSD model has two important limitations.  First, Figure~\ref{fig:locerra} shows that the Savin-Doyle model has little ability to correct negatively biased MSDs at the shortest timescales.  Indeed, the camera aperture time $\tau$ is typically at least an order of magnitude smaller than $\dt$, in which case the effect of the dynamic error in Figure~\ref{fig:locerra} is extremely small, and insufficient to explain larger negative MSD biases as in Figure~\ref{fig:controla}.  Second, the Savin-Doyle model uses one parameter ($\tau$) to lower %\correct{greg don't like the word diminish, but I cannot find a good substitute..}
the MSD, and a different parameter ($\sigma$) to raise it.  This leads to an identifiability issue which adversely affects the subdiffusion estimator, as we shall see in Section~\ref{sec:sim}.  Complementing the theoretically derived Savin-Doyle approach, we present a general high-frequency noise filtering framework in the following section.  

%-------------------------------------------------------------------------------

\section{Proposed Method}\label{sec:filter}

In order to formulate our proposed method of filtering the localization errors in single particle tracking experiments, we begin with the following definition of high frequency noise.  Let us first focus on a
one-dimensional zero-drift CSI process $X(t)$ with $E[X(t)] = 0$,
% and as in Section~\ref{sec:sdmodel},
and let $\Xts = \{X_n: n \ge 0\}$ and $\Yts = \{Y_n: n \ge 0\}$ denote the true and recorded particle position process at times $t = n \cdot \dt$. 
% as in Section~\ref{sec:sdmodel}.
Then we shall say that the observation process $\Yts$ contains only high frequency noise if the low-frequency second-order dynamics of the true and recorded particle positions are the same, namely
\begin{equation}\label{hypo:msd}
  \lim_{n\rightarrow\infty} \frac{\MSD_{Y}(n)}{\MSD_{X}(n)} = 1.
\end{equation}
Given the true position process $\Xts$, our noise model sets the observed position process to be of autoregressive/moving-average $\armapq$ type:
\begin{equation}\label{eq:farma}
  Y_n = \sum_{i=1}^p \theta_i Y_{n-i} + \sum_{j=0}^q \rho_j X_{n-j}, \qquad n \ge r = \max\{p,q\}.  
\end{equation}

For $0 \le n < r$, $Y_n$ is defined via the stationary increment process $\Delta \Xts = \{\Delta X_n: n \in \mathbb Z\}$.  That is, with the usual parameter restrictions
\begin{equation}\label{eq:stat}
  \min_{\{z \in \mathbb C: |z| \le 1\}} {\textstyle \big\vert 1-\sum_{i=1}^p \theta_i z^i\big\vert} > 0, \qquad \min_{\{z \in \mathbb C: |z| \le 1\}} {\textstyle \big\vert \rho_0 - \sum_{j=1}^q \rho_j z^j\big\vert} > 0,
\end{equation}
\citep[e.g.][]{brockwell.davis91}, the increment process $\Delta \Yts = \{\Delta Y_n: n \in \mathbb Z\}$ defined by
\begin{equation}\label{eq:armastat}
  \Delta Y_n = \sum_{i=1}^p \theta_i \Delta Y_{n-i} + \sum_{j=0}^q \rho_j \Delta X_{n-j}
\end{equation}

is a well-defined stationary process which can be causally derived from $\Delta \Xts$, and vice-versa.  Moreover, setting $Y_n = \sum_{i=0}^{n-1} \Delta Y_i$ obtains the ARMA relation~\eqref{eq:farma}  on the position scale for $n \ge r$.

One may note in model~\eqref{eq:farma} that $\rrh = (\rv [0] \rho q)$ and $\var(\Delta X_n)$ cannot be identified simultaneously.  This issue is typically resolved in the time-series literature by imposing the restriction $\rho_0 = 1$.  However, in order for the recorded positions to adhere to a high-frequency error model as defined by~\eqref{hypo:msd}, a different restriction must be imposed:
\begin{theorem}\label{thm:restriction}
  Let $\Xts$ and $\Yts$ denote the true and recorded position processes, with the latter defined by an $\armapq$ representation of the former as in~\eqref{eq:armastat}.  Then $\Yts$ is a high-frequency error model for $\Xts$ as defined by~\eqref{hypo:msd} if and only if
  \begin{equation}\label{eq:armarestr}
    \rho_0 = 1 - \sum_{i=1}^p \theta_i - \sum_{j=1}^q \rho_j.
  \end{equation}
\end{theorem}
The proof is given in Appendix~\ref{appendix:thm1}.  Indeed, the following result (proved in Appendix~\ref{appendix:thm2}) shows that the family of $\armapq$ noise models~\eqref{eq:farma} is sufficient to describe any high-frequency noise model to arbitrary accuracy:
\begin{theorem}\label{thm:hfapprox}
  Let $\Yts$ be a stochastic process of recorded positions defined as a high-frequency noise model via~\eqref{hypo:msd}.  If $\Yts$ satisfies the assumptions in Appendix~\ref{appendix:thm2}, then for any $\epsilon > 0$ we may find an $\armapq$ noise model $\bm{\mathcal{Y}^\star} = \{Y^\star_n: n \ge 0\}$ satisfying~\eqref{eq:farma} such that for all $n \ge 0$ we have
  \begin{equation}
    \left\vert\frac{ \MSD_{Y^\star}(n)}{\MSD_{Y}(n)  } -1 \right\vert < \epsilon.	
  \end{equation}
\end{theorem}

\subsection{Efficient Computations for the Location-Scale Model}

Let us now consider a $\nd$-dimensional position process $\X(t) = \sum_{j=1}^\np \bbe_j f_j(t) + \SSi^{1/2} \Z(t)$ following the location-scale model~\eqref{eq:lsmodel}.  Then we may construct an $\armapq$ high-frequency model for the measured positions as follows.  Starting from the drift-free stationary increment process $\Delta \tilde \Xts = \{\Delta \tilde \X_n = \SSi^{1/2}\Delta \Z_n: n \in \mathbb Z\}$, define the increment process $\Delta \tilde \Yts = \{\Delta \Y_n: n \in \mathbb Z\}$ via
\begin{equation}\label{eq:hqzero}
  \Delta \tilde \Y_n = \sum_{i=1}^p \theta_i \Delta \tilde \Y_{n-i} + \sum_{j=0}^q \rho_j \Delta \tilde \X_{n-j}.
\end{equation}
Then under parameter restrictions~\eqref{eq:stat}, $\Delta \tilde \Yts$ is a well-defined stationary process with $E[\Delta \tilde \Y_n] = \bm 0$.  In order to add drift to the high-frequency noise model~\eqref{eq:hqzero}, let
\begin{equation}\label{eq:hqdrift}
  \begin{aligned}
    \dX_n & =
    \begin{cases}
      \Delta \tilde \X_n, & n < 0, \\
      \Delta \tilde \X_n + \sum_{m=1}^d \bbe_j \Delta f_{nj}, & n \ge 0,
    \end{cases} \\
    \Delta \Y_n & =
    \begin{cases} \Delta \tilde \Y_n, & n < 0 \\
      \sum_{i=1}^p \theta_i \dY_{n-i} + \sum_{j=0}^q \rho_j \dX_{n-j}, & n \ge 0,
    \end{cases}
  \end{aligned}
\end{equation}
where $\Delta f_{nj} = f_j((n+1)\cdot \dt) - f_m(n\cdot \dt)$.  Then for $n \ge 0$, $\X_n = \sum_{i=0}^{n-1} \dX_i$ corresponds to discrete-time observations of $\X(t)$ from the location-scale model~\eqref{eq:lsmodel}, and $\Y_n = \sum_{i=0}^{n-1} \dY_i$ satisfies the $\armapq$ relation~\eqref{eq:farma}.  Moreover, the observed increments $\dY = (\rv [0] \dY {N-1})$ follow a matrix-normal distribution
\[
  \dY \sim \MN(\F_\pph\bbe, \VV_\pph, \SSi),
\]
where $\F_\pph$ is an $N \times k$ matrix with elements 
\begin{equation}\label{eq:fma-drift}
  F_{nm} = -\sum_{i=1}^{\min\{n,p\}}\theta_i F_{n-i,m} + \sum_{j=0}^{\min\{n,q\}}\rho_j \Delta f_{n-j,m},
\end{equation}
and $\VV_\pph$ is an $N \times N$ Toeplitz matrix with element $(n,m)$ given by $V_\pph^{(n,m)} = \acf_{\Delta Y}(|n-m|)$. Thus, we may use the computationally efficient methods of Section~\ref{sec:parest} for parameter inference, given the autocorrelation function $\acf_{\Delta Y}(n)$ defined by~\eqref{eq:armastat}.   For pure moving-average processes ($p = 0$), this function is available in closed-form given an arbitrary true increment autocorrelation function $\acf_{\Delta Z}(n)$.  For $p > 0$, an accurate and computationally efficient approximation is provided in Appendix~\ref{appendix:farma-acf}.

\subsection{The Fractional $\ma(1)$ Noise Model}\label{sec:fma}

Perhaps the simplest $\armapq$ noise model is that with $p = 0$ and $q = 1$, i.e., the first-order moving-average $\mao$ model given by
\begin{equation}\label{eq:fma1}
  \Y_n = (1-\rho) \X_n + \rho \X_{n-1},
\end{equation}
where $|\rho| < 1$ is required to satisfy~\eqref{eq:stat}, and $\rho < \tfrac 1 2$ is required to satisfy~\eqref{hypo:msd}.  The autocorrelation of the observed increments becomes
\begin{equation}\label{eq:ma-acf}
  \acf_{\dY}(n) = \acf_{\dX}(n) + (1-\rho)\rho \big[ \acf_{\dX}(|n-1|) + \acf_{\dX}(n+1) - 2\acf_{\dX}(n) \big],
\end{equation}
where $\acf_{\dX}(n)$ is the autocorrelation of the true increment process.  Of particular interest is when $\X(t)$ is fractional Brownian motion, for which we refer to the corresponding $\mao$ noise model as fMA. The MSD of such a model is plotted in Figure~\ref{fig:fma-msda} for a range of values $\rho \in (-1, \tfrac 1 2)$. 
As with the fractional Savin-Doyle (fSD) model~\eqref{eq:locmod}, %described in Section~\ref{sec:sdmodel},
$\rho > 0$ raises the high-frequency correlations in the observation process, whereas $\rho < 0$ lowers them.  A similar MSD plot for the fSD model is given in Figure~\ref{fig:fma-msdb}.
While both high-frequency noise models can similarly raise the MSD at short timescales, the fMA model has much higher capacity to lower it.
\begin{figure}[htbp!]
  \centering
  \begin{subfigure}{\textwidth}
    \includegraphics[width = 1\textwidth]{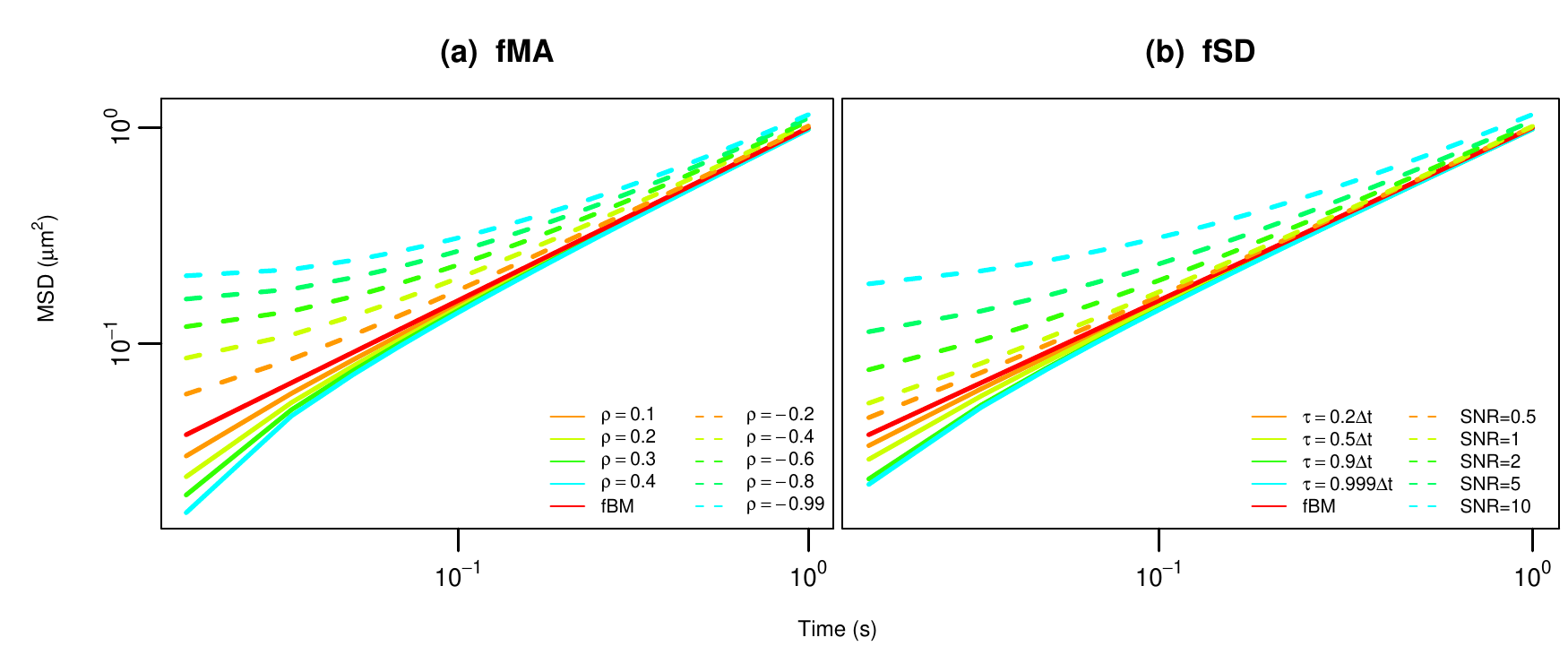}
    \phantomsubcaption\label{fig:fma-msda}\phantomsubcaption\label{fig:fma-msdb}
  \end{subfigure}
  \caption[MSD of fMA and fSD models]{(a) MSD of the fMA model with $\alpha=0.8$ and different values of $\rho$. 
    (b) MSD of the fSD model with $\alpha=0.8$ and different values of $\tau$ and signal-to-noise ratio
    % $\sigma^2$. Signal-to-noise ratio (SNR) is defined as SNR =
    $\SNR = \var(\Delta B^\alpha) / \sigma^2$.}
  \label{fig:fma-msd}
\end{figure}

In order to examine this difference more carefully, the following experiment is proposed.  Suppose that observed increments $\dY = (\rv [0] {\dY} {N-1})$ are generated from a drift-free location-scale fSD model $p(\dY \mid \alpha, \SSi, \tau, \sigma)$.  Then for fixed $N$ and $\dt$, we may calculate the parameters of the (drift-free) fMA model $p(\dY \mid \alpha_\star, \SSi_\star, \rho)$ which minimize the Kullback-Liebler divergence from the true model,
\begin{align*}
  (\hat \alpha_\star, \hat \SSi_\star, \hat \rho)
  & = \argmin_{(\alpha_\star, \SSi_\star, \rho)} \textnormal{KL}\big\{p(\dY \mid \alpha, \SSi, \tau, \sigma)\, \| \, p(\dY \mid \alpha_\star, \SSi_\star, \rho)\big\} \\
  % & = \argmin_{(\alpha_\star, \SSi_\star, \rho)} \frac 1 2 \times \left(\tr(\SSi_\star^{-1}\SSi)\tr(\VV_\star^{-1}\VV) + \log\left(\frac{|\SSi_\star|^N|\VV_\star|^k}{|\SSi|^N|\VV|^k}\right) - Nk \right),
  & = \argmin_{(\alpha_\star, \SSi_\star, \rho)} \tr(\SSi_\star^{-1}\SSi)\tr(\VV_\star^{-1}\VV) + \log\left(\frac{|\SSi_\star|^N|\VV_\star|^k}{|\SSi|^N|\VV|^k}\right),
  % & = \argmin_{(\alpha_\star, \bbe_\star, \SSi_\star, \rho)} \frac{1}{2} \left( \tr(\OOm_\star^{-1} \OOm) + (\Psi_\star - \Psi)^T \OOm_\star^{-1} (\Psi_\star - \Psi) + \ln( \frac{|\OOm_\star|}{|\OOm|} ) - Nk  \right).
\end{align*}
where $\VV$ and $\VV_\star$ are $N \times N$ Toeplitz variance matrices with first row given by the autocorrelation function of the fSD and fMA models, respectively.

Figure~\ref{fig:costa} displays the difference between true and best-fitting subdiffusion parameters $\hat \alpha_\star - \alpha$ and $\log \hat D_\star - \log D$, for $k=2$, $\SSi = \left[\begin{smallmatrix} 1 & 0 \\ 0 & 1 \end{smallmatrix}\right]$, $N = 1800$, $\dt = 1/60$, and over a range of parameter values $(\alpha, \tau, \sigma)$.  Figure~\ref{fig:costb} does the same, but with the best-fitting fSD model to data generated from fMA.  For all but very high static error $\sigma$ (corresponding to low signal-to-noise ratio $\textnormal{SNR} = \var(\Delta X_n)/\sigma^2$), the fMA model can recover the true subdiffusion parameters $\aD$ with little bias due to model misspecification.  There is significantly more bias when fSD is used on data generated from fMA, particularly when $\rho > 0$ as suggested by Figure~\ref{fig:fma-msd}.
\begin{figure}[htbp!]
  \centering
  \begin{subfigure}{\textwidth}
    \includegraphics[width = 1\textwidth]{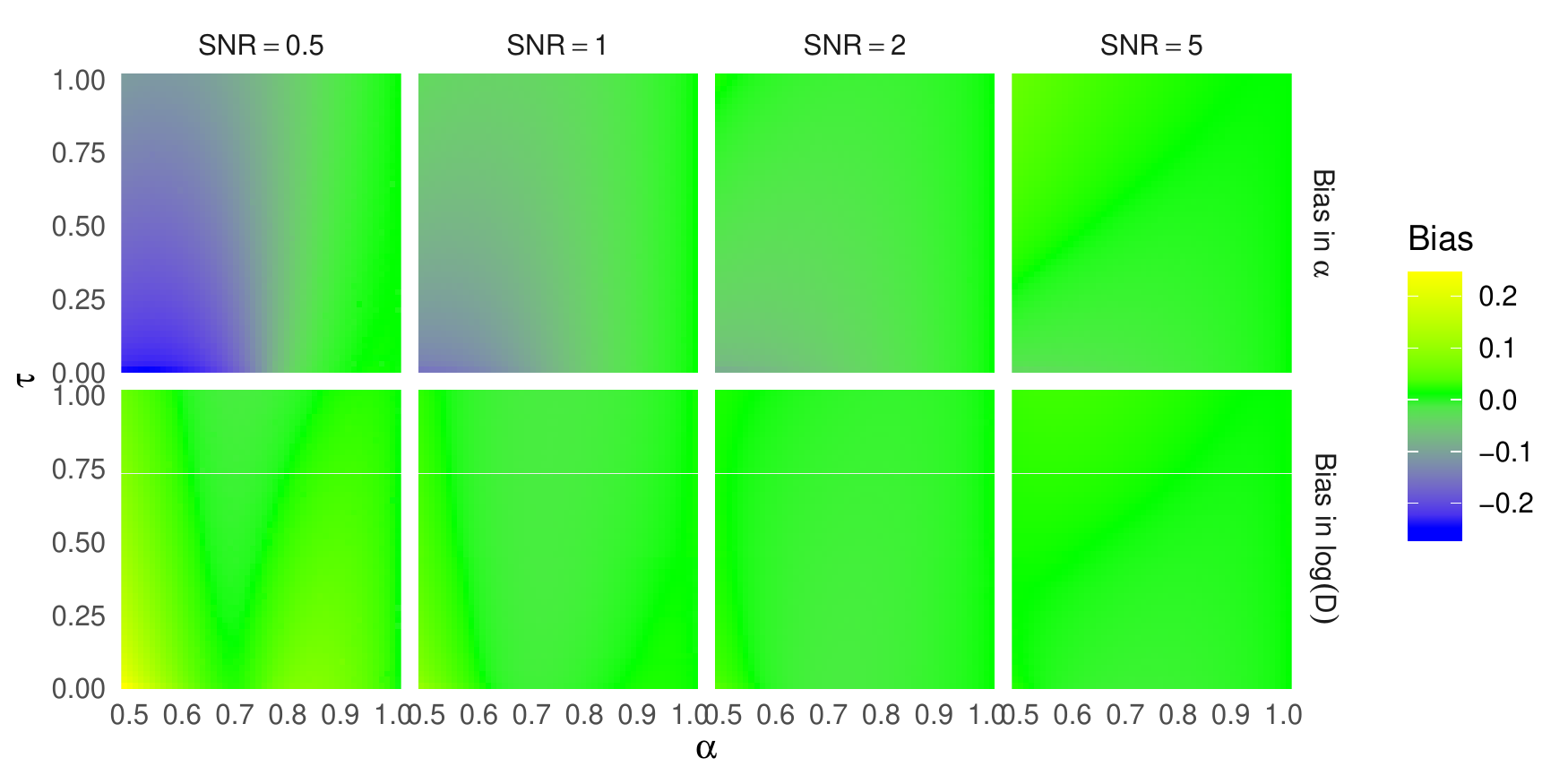}
    \includegraphics[width = 1\textwidth]{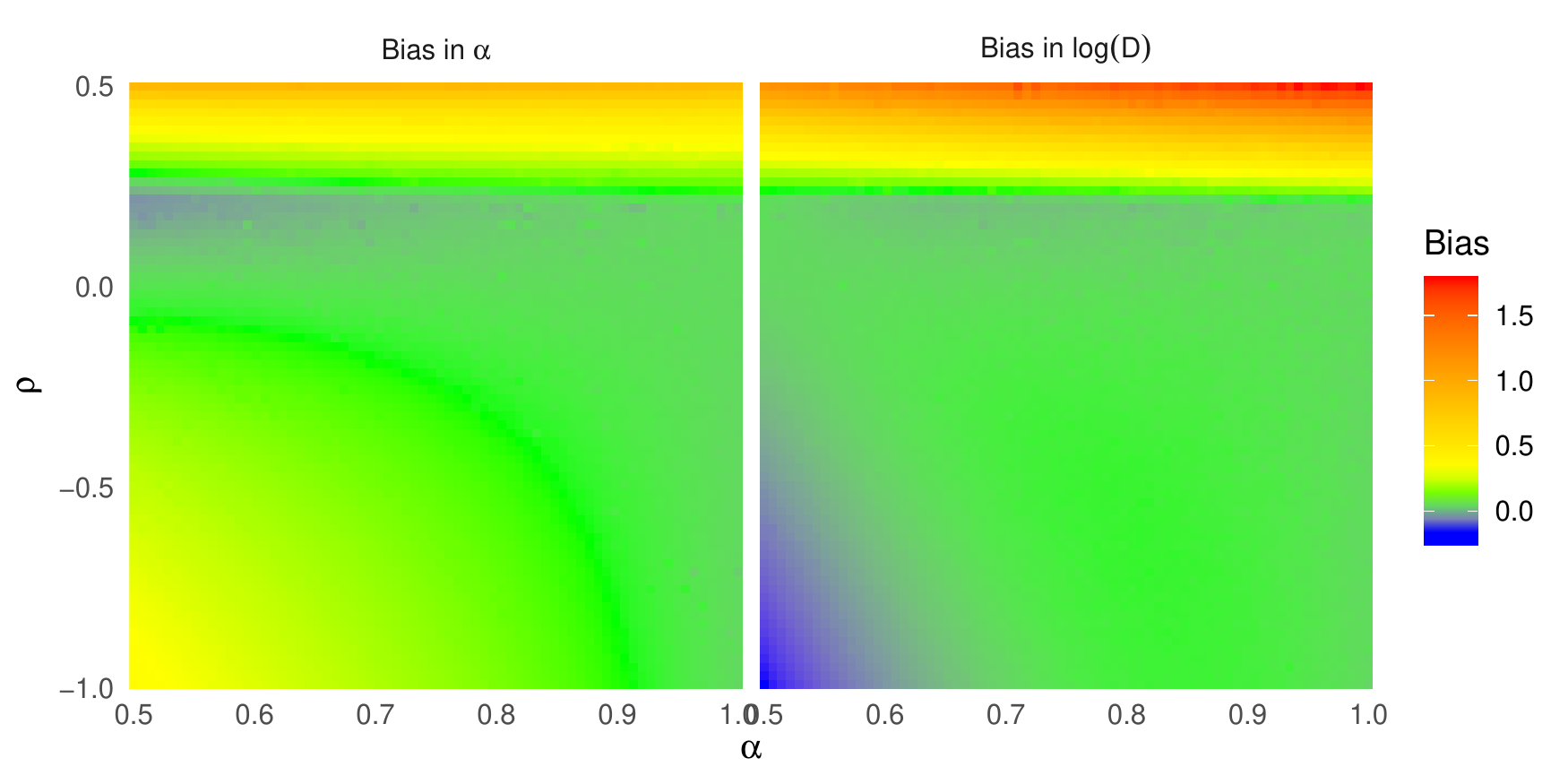}
    \phantomsubcaption\label{fig:costa}\phantomsubcaption\label{fig:costb}
  \end{subfigure}
  \caption{Model misspecification bias in $\alpha$ and $D$.
    (a) Best-fitting fMA model to true fSD models with different values of $\alpha$,  $\tau$, and signal-to-noise ratio $\SNR = \var(\Delta B_n^\alpha) / \sigma^2$.
    (b) Best-fitting fSD model to true fMA models with different values of $\alpha$ and $\rho$.}
  \label{fig:cost}
\end{figure}

%-------------------------------------------------------------------------------

\section{Simulation Study}\label{sec:sim}

In this section, we evaluate the performance of the proposed $\armapq$ high-frequency noise filters in various simulation settings.  In each setting, we simulate $B = 500$ observed data trajectories $\Y^{(b)} = (\rv [0] {\Y^{(b)}} N)$, $b = 1,\ldots,B$, each consisting of $N = 1800$ two-dimensional observations ($k = 2$) recorded at intervals of $\dt = \SI{1/60}{\second}$.

\subsection{Empirical Localization Error}\label{sec:sim-loc}

Consider the following simulation setting designed to reflect the localization errors in our own experimental setup.  Let $\Yv$ denote the trajectory measurements for a particle undergoing ordinary diffusion in a viscous environment.  Then we may estimate the MSD ratio
\begin{equation}\label{eq:gn}
g(n) = \frac{\MSD_{\tYv}(n)}{\MSD_{\Xv}(n)},
\end{equation}
where the MSD of the true position process is $\MSD_{\Xv}(n) = 2D t$ with $D$ determined by the Stokes-Einstein relation~\eqref{eq:stokes}, and the MSD of the drift-subtracted observation process $\tYv$ can be accurately estimated by
\[
\widehat{\MSD}_{\tYv}(n) = \frac{1}{M} \sum_{i=1}^M \widehat{\MSD}_{\tYv^{(i)}}(n),
\]
where $\widehat{\MSD}_{\tYv^{(i)}}(n)$ is the empirical MSD~\eqref{eq:msdemp} for each (drift-subtracted) particle trajectory $\tYv^{(1)}, \ldots, \tYv^{(M)}$ recorded in a given experiment (e.g., Figure~\ref{fig:controla}).  We then suppose that the true trajectory is drift-free fBM $\X(t) = \SSi^{1/2} \BB^\alpha(t)$, and simulate the measured trajectories from
\[
\Y^{(b)} \iid \MN\left(\bm 0, \VV, \SSi\right),
\]
where $\SSi = \left[\begin{smallmatrix} 1 & 0 \\ 0 & 1 \end{smallmatrix}\right]$ and the $(N+1) \times (N+1)$ variance matrix $\VV$ is that of a CSI process with MSD given by
\begin{equation}\label{eq:msdg}
\MSD_{\Y}(n) = (\gamma\hat g(n) - \gamma + 1) \times \MSD_{\X}(n),
\end{equation}
where $\hat g(n)$ is the estimated noise ratio~\eqref{eq:gn} from a viscous experiment, and the noise factor $\gamma > 0$ can be used to suppress or amplify the empirical localization error with $\gamma < 1$ or $\gamma > 1$, respectively.  Having constrained our estimator such that $\hat g(n) = 1$ for $n > N_0$,~\eqref{eq:msdg} is a high-frequency noise model as defined by~\eqref{hypo:msd}.  Figure~\ref{fig:sim-demo} displays the observed MSD~\eqref{eq:msdg} for a true fBM trajectory with $\alpha = 0.6$, contaminated by empirical localization errors from two representative viscous experiments described in Table~\ref{tab:data}, illustrating the effects of high-frequency MSD suppression and amplification, respectively.
\begin{figure}[htbp!]
  \centering
  \includegraphics[width = 1\textwidth]{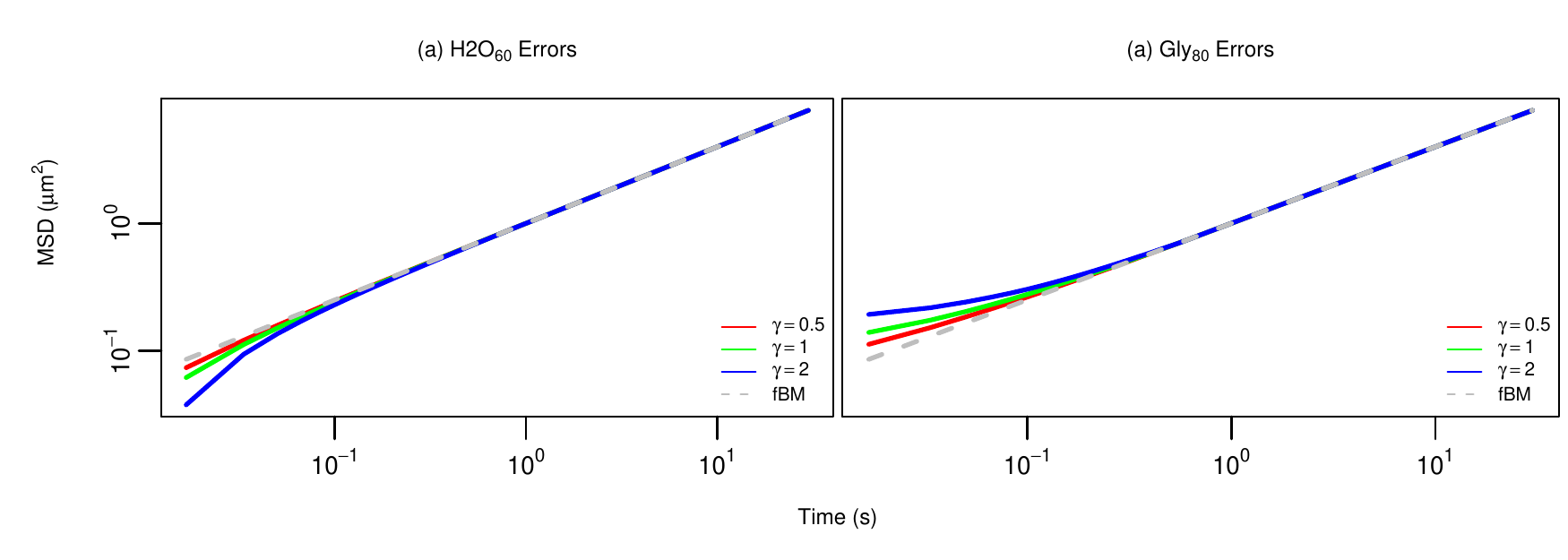}
  \caption[experiments for simulation]{MSD of simulated observations with empirical localization error~\eqref{eq:msdg}, where the true trajectory is an fBM process with $\alpha = 0.6$.  (a) High-frequency MSD suppression as observed in $\water_{60}$ experiment (see Table~\ref{tab:data}).
    (b) High-frequency MSD amplification as observed in $\gly_{60}$ experiment.}
  \label{fig:sim-demo}
\end{figure}

The following methods are used to estimate the subdiffusion parameters $\aD$ for each set of simulated particle observations $\Y^{(b)}$, $b = 1,\ldots B$:
\begin{enumerate}
\item \textbf{LS:} The semiparameteric least-squares estimator~\eqref{eq:ols} applied to the drift-subtracted empirical MSD~\eqref{eq:msdemp}.
\item \textbf{fBM:} The MLE of an fBM-driven location-scale model with linear drift,
  \begin{equation}\label{eq:fbmlin}
    \X(t) = \mmu t + \SSi^{1/2} \BB^\alpha(t),
  \end{equation}
  for which the model parameters are $(\alpha, \mmu, \SSi)$.
\item \textbf{fSD:} The MLE of the Savin-Doyle error model~\eqref{eq:locmod} applied to~\eqref{eq:fbmlin}, for which the model parameters are $(\alpha, \tau, \sigma, \mmu, \SSi)$.
\item \textbf{fMA:} The MLE of the proposed $\ma(1)$ high-frequency noise filter~\eqref{eq:fma1} applied to~\eqref{eq:fbmlin}, for which the model parameters are $(\alpha, \rho, \mmu, \SSi)$.
\item \textbf{fMA2:} The MLE of the proposed $\ma(2)$ high-frequency noise filter
  \[
    \Y_n = (1-\rho_1-\rho_2) \X_n + \rho_1 \X_{n-1} + \rho_2 \X_{n-2}
  \]
  applied to~\eqref{eq:fbmlin}, for which the model parameters are $(\alpha, \rho_1, \rho_2, \mmu, \SSi)$.
\item \textbf{fARMA:} The MLE of the proposed $\arma(1,1)$ high-frequency noise filter
  \[
    \Y_n = \theta \Y_{n-1} + (1-\theta-\rho) \X_n + \rho \X_{n-1}
  \]
  applied to~\eqref{eq:fbmlin}, for which the model parameters are $(\alpha, \theta, \rho, \mmu, \SSi)$.
\end{enumerate}
\begin{remark}
  The fSD exposure time parameter $\tau$ is typically known and therefore need not be estimated from the data.  However, we have opted here to estimate it regardless, as this gives far greater ability to account for high-frequency MSD suppression (e.g., Figure~\ref{fig:locerra}).  We return to this point in Section~\ref{sec:exper}.
\end{remark}

The point estimates for $\aD$ for true fBM trajectories with $\alpha \in \{.6, .8, 1\}$ and empirical error factor $\gamma \in \{.5, 1, 2\}$ are displayed in Figure~\ref{fig:pos-est}. 
\begin{figure}[htbp!]
  \centering
  \includegraphics[width = \textwidth]{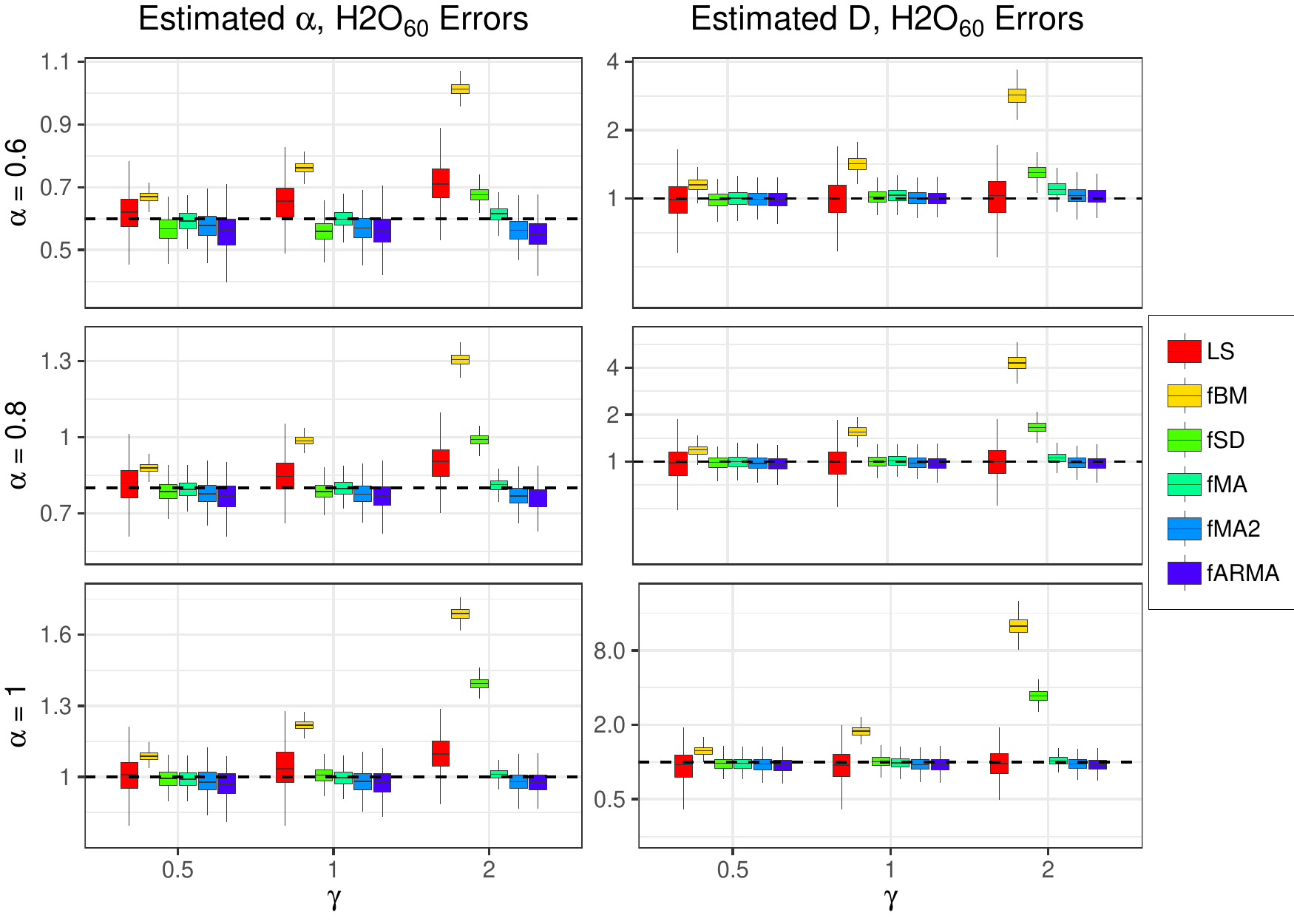}
  \includegraphics[width = \textwidth]{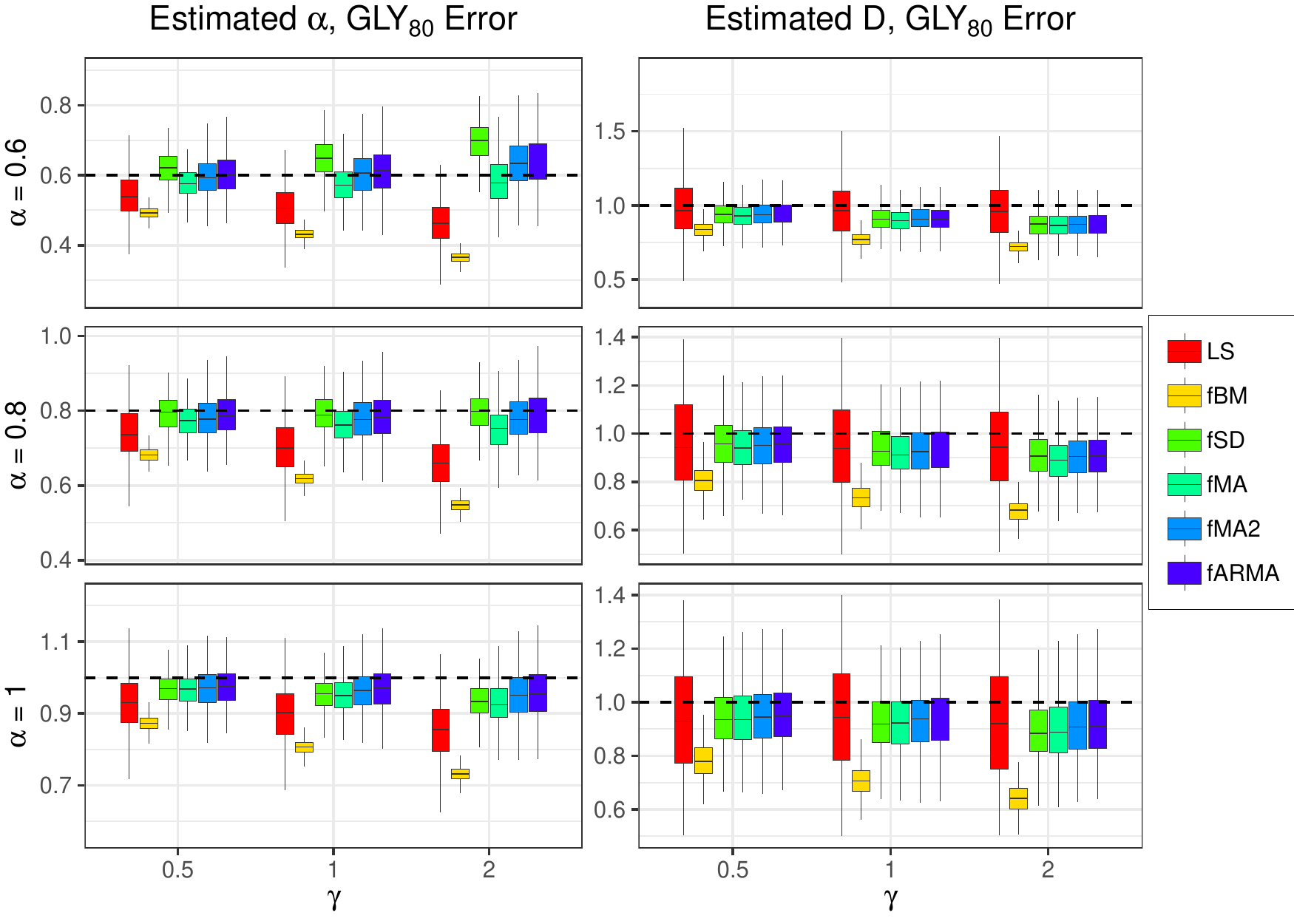}
  \caption{Estimates of $\aD$ for true fBM trajectories with various types and degrees of empirical localization errors.}
  \label{fig:pos-est}
\end{figure}
\begin{table}[!htb]
  \caption{Actual coverage by 95\% confidence intervals with various types and degrees of empirical localization errors.}
  \centering
  \begin{tabular}{@{}lrcccccccc@{}}
    \toprule
    \multicolumn{2}{c}{\multirow{2}{*}{$\rcov(\alpha)$}}
    % \multicolumn{2}{c}{}
    & %\phantom{abc}
    & \multicolumn{3}{c}{$\water_{60}$ Errors}
    & %\phantom{abc}
    & \multicolumn{3}{c}{$\gly_{80}$ Errors} \\
    \cmidrule{4-6} \cmidrule{8-10}
    &
    % \multicolumn{2}{c}{$\rcov(\alpha)$}
    && $\gamma=0.5$ & $\gamma=1$ & $\gamma=2$ && $\gamma=0.5$ & $\gamma=1$ & $\gamma=2$ \\ 
    \midrule
    \multirow{5}{*}{$\alpha = 0.6$}
    &fBM && 5 &0 &0 & & 0 & 0 & 0 \\ 
    &fSD  && 90  & 87  & 11 && 93 & 84 & 59\\
    &fMA  && 96  & 96   & 90 && 91 & 88 & 88\\
    &fMA2  && 91  & 91   & 84 && 94 & 95 & 94\\
    &fARMA  && 92  & 93   & 87 && 89 & 93 & 93\\
    \midrule
    \multirow{5}{*}{$\alpha = 0.8$}
    &fBM  && 4 & 0 & 0 && 0  & 0  & 0 \\ 
    &fSD  && 91  & 93  & 0 && 92 & 94 & 94\\
    &fMA  &&  93  & 94   & 93 && 87 & 84 & 81\\
    &fMA2  &&  93  & 91   & 87 && 92 & 91 & 93\\
    &fARMA  &&  92  & 91   & 88 && 89 & 90 & 93\\
    \midrule
    \multirow{5}{*}{$\alpha = 1$}
    &fBM   && 1 & 0 & 0 && 0  & 0   & 0 \\ 
    &fSD  &&  13  & 6   & 0 && 23 & 34 & 36\\
    &fMA   &&  95  & 94  & 93 && 87 & 81 & 70\\
    &fMA2   &&  92  & 92  & 94 && 90 & 88 & 84\\
    &fARMA   &&  91  & 92  & 92 && 87 & 86 & 85\\
    \midrule
    \multicolumn{2}{c}{\multirow{2}{*}{$\rcov(\log D)$}}
    % \multicolumn{2}{c}{}
    & %\phantom{abc}
    & \multicolumn{3}{c}{$\water_{60}$ Errors}
    & %\phantom{abc}
    & \multicolumn{3}{c}{$\gly_{80}$ Errors} \\
    \cmidrule{4-6} \cmidrule{8-10}
    &
    % \multicolumn{2}{c}{$\rcov(D)$}
    && $\gamma=0.5$ & $\gamma=1$ & $\gamma=2$ && $\gamma=0.5$ & $\gamma=1$ & $\gamma=2$ \\ 
    \midrule
    \multirow{5}{*}{$\alpha = 0.6$}
    &fBM   && 57 &1 &0 & & 20 & 1 & 0 \\ 
    &fSD  && 94   & 96  & 10 && 88 & 80 & 72\\
    &fMA  && 96  & 95   & 88 && 86 & 73 & 85\\
    &fMA2  && 94  & 95   & 95 && 86 & 79 & 66\\
    &fARMA  && 94  & 95   & 95 && 87 & 79 & 65\\
    \midrule
    \multirow{5}{*}{$\alpha = 0.8$} &
                                      fBM  && 48 & 0 & 0 && 18  & 2 & 0 \\ 
    &fSD  && 92  & 94  & 1 && 90 & 89 & 82\\
    &fMA  &&  95  & 94   & 94 && 89 & 82 & 76\\
    &fMA2  &&  93  & 94   & 94 && 89 & 86 & 83\\
    &fARMA  &&  91  & 93   & 93 && 89 & 88 & 84\\
    \midrule
    \multirow{5}{*}{$\alpha = 1$}
    &fBM   && 42 & 0 & 0 && 16  & 1   & 0 \\ 
    &fSD  &&  63  & 61   & 0 && 69 & 74 & 67\\
    &fMA   &&  95  & 94  & 95 && 90 & 88 & 80\\
    &fMA2  &&  92  & 92  & 94 && 91 & 90 & 85\\
    &fARMA  &&  90  & 91  & 93 && 91 & 89 & 85\\
    \bottomrule
  \end{tabular}
  \label{tab:pos-neg}
\end{table}
As expected, the semiparametric LS estimator is substantially more variable than any of the fully parametric estimators, and the error-unadjusted fBM estimator incurs considerable bias, even with the smallest noise factor $\gamma = 0.5$.  The high-frequency estimators (fMA, fMA2, and fARMA) are fairly similar to each other, with the additional parameters of fMA2 and fARMA giving them slightly lower bias and higher variance.  The high-frequency estimators are slightly more biased than fSD in the $\gly_{80}$ simulation with $\alpha = 0.8$.
% Upon closer inspection, it was found that this setting corresponded almost exactly to Savin-Doyle purely static errors $\Y_n = \X_n + \eps_n$ with $\SNR \in \{5.6, 2.7, 1.2\}$.
In contrast, they are somewhat less biased than fSD for $\gly_{80}$ with the stronger subdiffusive signal $\alpha = 0.6$, and considerably less so for $\water_{60}$ with the largest noise factor $\gamma = 2$.

Table~\ref{tab:pos-neg} displays the true coverage of the 95\% confidence intervals for each parametric estimator, calculated as
\[
\rcov(\psi) = \frac{1}{B} \sum_{b=1}^B \mathfrak 1\{\theta \in \hat\psi_b \pm 1.96 \se(\hat\psi_b)\},
\]
where $\psi \in \{\alpha, \log D\}$, $\hat \psi_b$ is the MLE for dataset $b$, and $\se(\hat\psi_b)$ is the square root of the corresponding diagonal element of the variance estimator $\widehat{\var}(\hat \tth_b) = -\left[ \frac{\partial^2 \ell(\Y^{(b)} \mid \hat \tth_b)}{\partial \tth \partial \tth'}\right]^{-1}$, where $\hat \tth_b$ is the MLE of all model parameters.  The true coverage of the fMA, fMA2, and fARMA confidence intervals is close to 95\% when the bias is negligible and typically above 85\%.  This is also true for fSD, with the notable exception of either empirical error model and true $\alpha = 1$.  Upon closer inspection, we found that the fSD model suffers from an identifiability issue in the diffusive (viscous) regime, wherein the MSD suppression by $\tau$ and amplification by $\sigma$ achieve the same net effect over a range of values.  This does not affect the estimate of $\aD$, but significantly decreases the curvature of $\ell(\Y \mid \hat \tth)$, thus artifically inflating the observed Fisher information $\widehat{\var}(\hat \tth_b)^{-1}$.

\begin{remark}
  Since the subdiffusion equation $\MSD_{\X}(t) = 2D t^\alpha$ dictates that $D$ be measured in units of $\si{\square\micro\meter\second}^{-\alpha}$, in order to compare estimates of $D$ for different values of $\alpha$ as in Figure~\ref{fig:pos-est}, we follow the convention of interpreting $D$ as half the MSD at time $t = \SI{1}{\second}$~\citep[e.g.,][]{lai.et.al07,wang.et.al08},  which for any $\alpha$ is measured uniformly in units of $\si{\square\micro\meter}$.
\end{remark}

\subsection{Modeling Transient Subdiffusion}\label{sec:gle}

In this section, we show how the proposed high-frequency filter can be used not only for measurement error correction, but also to estimate subdiffusion in models where the power-law relation $\MSD_{\X}(t) \sim t^{\alpha}$ holds only for $t > \tmin$.  For this purpose, here we shall generate particle trajectories from a so-called \emph{Generalized Langevin Equation} (GLE), a physical model derived from the fundamental laws of thermodynamics for interacting-particle systems~\citep[e.g.,][]{kubo66,zwanzig01,kou08}.   For a one-dimensional particle with negligible mass, the GLE for its trajectory $X(t)$ is a stochastic integro-differential equation of the form
\begin{equation}\label{eq:gle}
\int_{-\infty}^t \phi(t-s) V(s) \ud s = F(t),
\end{equation}
where $V(t) = \frac{\ud}{\ud t}X(t)$ is the particle velocity, $\phi(t)$ is a memory kernel, and $F(t)$ is a stationary mean-zero Gaussian force process with \mbox{$\acf_F(t) = \kB T \cdot \phi(t)$}, where $T$ is temperature and $\kB$ is Boltzmann's constant.
% The precise relationship between the memory kernel and the autocorrelation of $F(t)$ is known as the Fluctuation-Dissipation theorem~\citep{kubo66}.
The memory of the process is modeled as a generalized Rouse kernel~\citep{mckinley.etal09}:
\begin{equation}\label{eq:rouse}
	\phi(t) = \frac{\nu}{K} \sum_{k=1}^{K} \exp(-|t| / \tau_k), \quad \tau_k = \tau \cdot (K / k)^{\gamma}.
\end{equation}
The sum-of-exponentials form of~\eqref{eq:rouse} is a longstanding linear model for viscoelastic relaxation~\citep[e.g.,][]{soussou.et.al70, ferry80,mason.weitz95}, whereas the specific parametrization of the relaxation modes $\tau_k$ has been shown for sufficiently large $K$ to exhibit \emph{transient} subdiffusion~\citep{mckinley.etal09},
\begin{equation}\label{eq:transub}
  \MSD_X(t) =
  \begin{cases}
    2 \Deff \cdot t^{\aeff} & \tmin < t < \tmax \\
    % t & \textnormal{otherwise},
    2 D_{\textnormal{min}} \cdot t & t < \tmin \\
    % 2 \Deff \cdot t^{\aeff}  & \tmin < t < \tmax \\
    2 D_{\textnormal{max}} \cdot t &  t > \tmax,
  \end{cases}
\end{equation}
where the subdiffusive range parameters $(\tmin, \tmax)$ and the effective subdiffusion parameters $(\aeff, \Deff)$ are implicit functions of $K$, $\gamma$, $\tau$, and $\nu$.  Details of the parameter conversions and 
the exact form of~\eqref{eq:transub} are provided in Appendix~\ref{appendix:GLE}.

Figure~\ref{fig:gle-msd} displays the MSD of various GLE processes with fixed $K = 300$, and $\{\gamma, \tau, \nu\}$ tuned to have $\aeff = 0.63$, $\Deff = 0.58$, and values of $\tmin/\dt = \{5,10,20,50,100\}$.  In all cases the value of $\tmax$ was several times larger than the experimental timeframe $N\dt = \SI{30}{\second}$, such that the observable MSD could potentially be matched by the fBM-driven high-frequency models of Section~\ref{sec:filter}. 
The trajectories for this experiment were simulated from
\[
  \Y^{(b)} \iid \MN(\bm 0, \VV, \SSi),
\]
where $\SSi = \left[\begin{smallmatrix} 1 & 0 \\ 0 & 1 \end{smallmatrix}\right]$ and $\VV$ is the $(N+1) \times (N+1)$ variance matrix of the GLE process~\eqref{eq:gle} with MSDs displayed in Figure~\ref{fig:gle-msd}.
\begin{figure}[htb!]
  \centering
  \includegraphics[width = .8\textwidth]{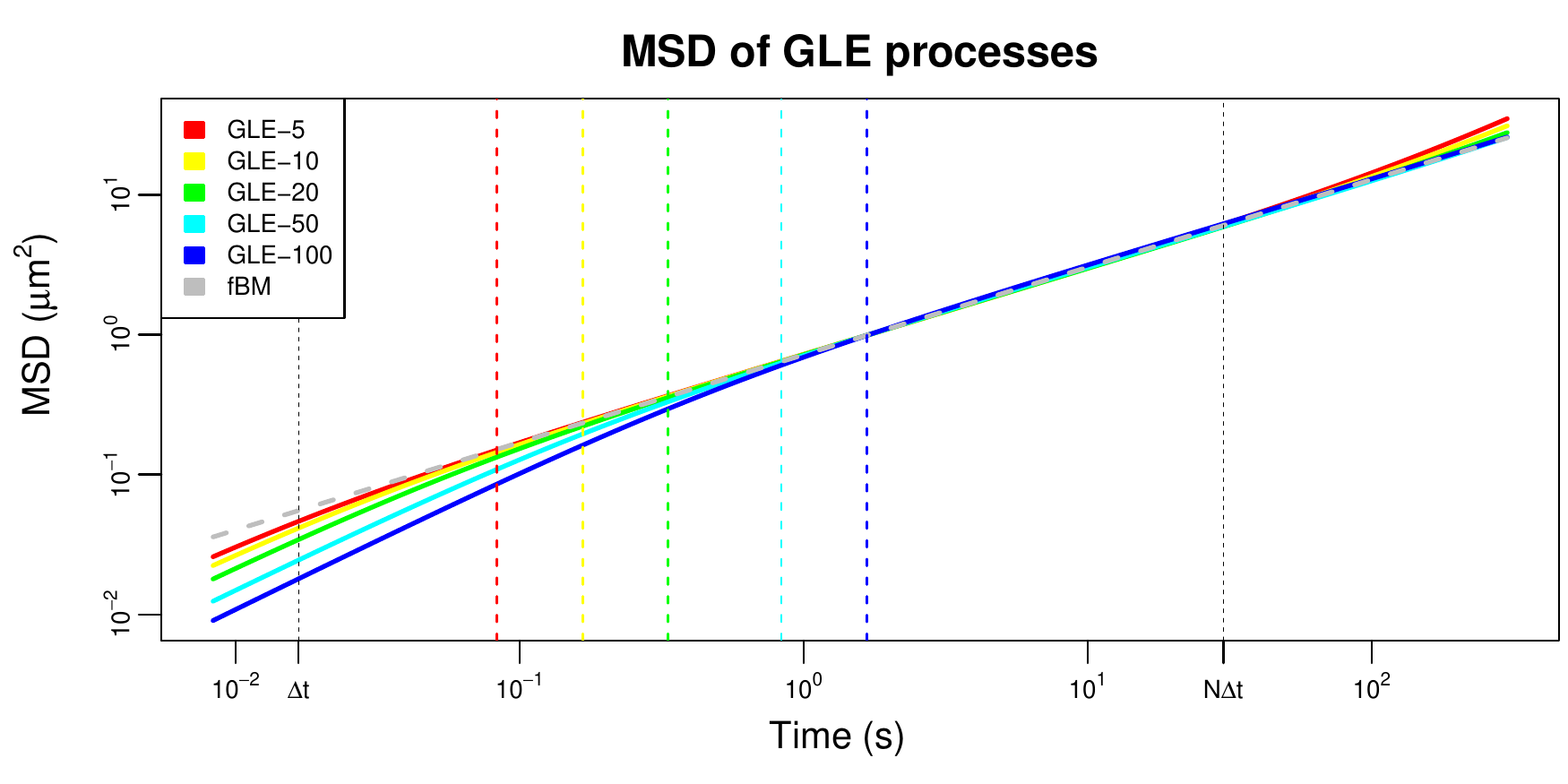}
  \caption{MSD of GLE processes with $\aeff = 0.63$, $\Deff = 0.58$, and $\tmin/\dt = \{5,10,20,50,100\}$.  The horizontal dashed lines indicated $\tmin$, and the diagonal dashed line corresponds to an fBM process with the same subdiffusive parameters $(\aeff, \Deff)$.
    The dotted vertical lines indicate the beginning and end of experiment, at $\dt = \SI{1/60}{\second}$ and $N\dt = \SI{30}{\second}$, respectively.}
  \label{fig:gle-msd}
\end{figure}

Figure~\ref{fig:gle-alpha} displays the parameter estimates of $\aeff$ and $\Deff$ for the six estimators described in Section~\ref{sec:sim-loc}, and Table~\ref{tab:gle} displays the true coverage probabilities of the corresponding 95\% confidence intervals.
\begin{figure}[htb!]
  \centering
  \includegraphics[width = \textwidth]{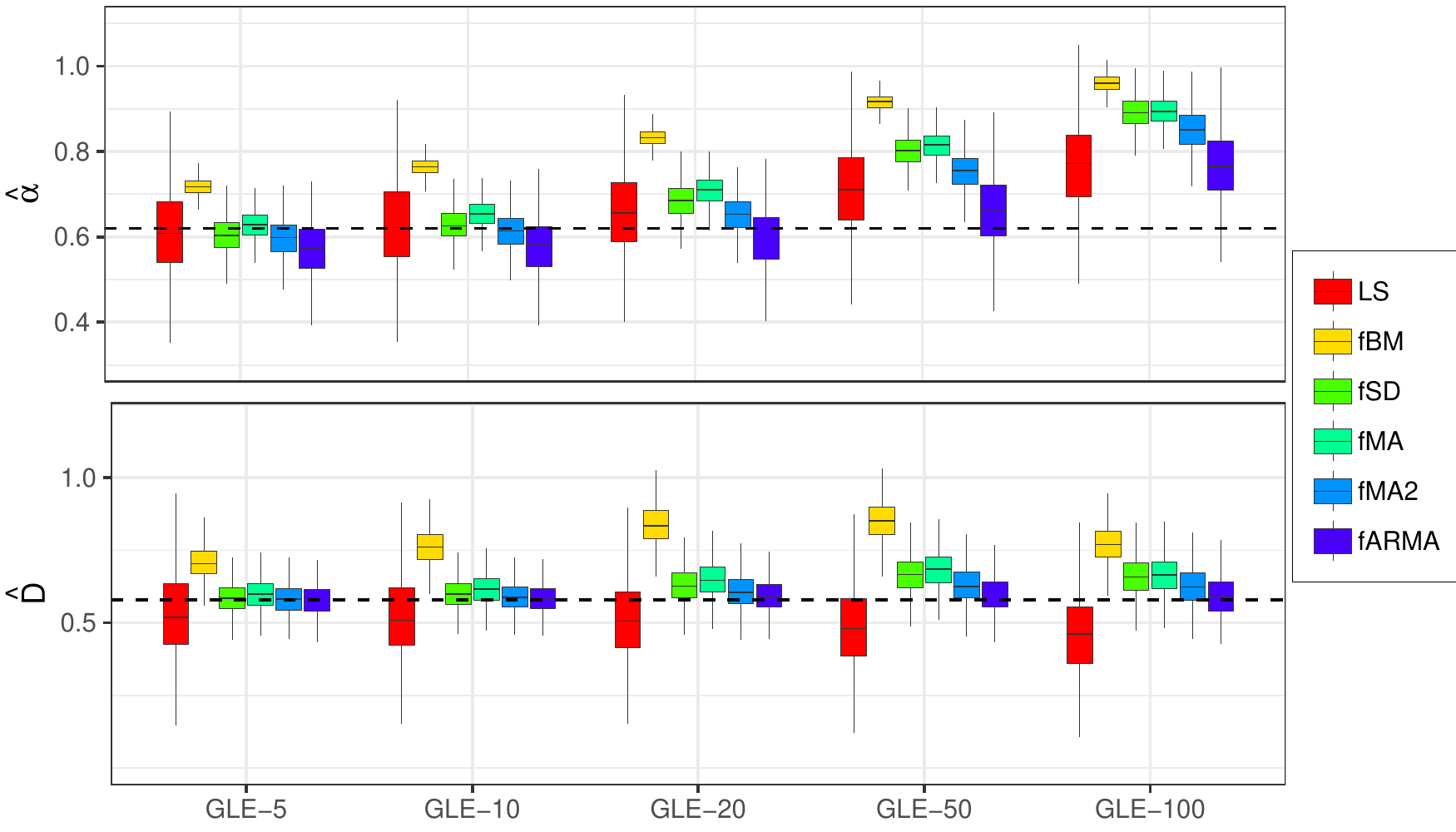}
  \caption{Estimates of $\aeff$ and $\Deff$ for simulated GLE trajectories with true parameters $\aeff = 0.63$, $\Deff = 0.58$, $K = 300$, and $\tmin/\dt = \{5,10,20,50,100\}$.}
  \label{fig:gle-alpha}
\end{figure}
\begin{table}[!htb]
\caption{Actual coverage by 95\% confidence intervals with different GLE processes.}
\centering
\begin{tabular}{@{}lccccc@{}}
\toprule
$\rcov(\alpha)$ &  GLE-5 & GLE-10 & GLE-20 & GLE-50 & GLE-100  \\ 
\midrule
fBM   &  0 & 0 & 0 & 0 & 0 \\ 
fSD   &  96 & 96& 64 & 0 & 0 \\
fMA   & 95 & 84 & 25 & 0 & 0 \\
fMA2   & 92 & 95 & 89 & 15 & 0 \\
fARMA   & 92 & 92& 95 & 85 & 53 \\
\midrule
$\rcov(\log D)$ &  GLE-5 & GLE-10 & GLE-20 & GLE-50 & GLE-100 \\ 
\midrule
fBM   &  31 & 8 & 1 & 1 & 11 \\ 
fSD   &  94 & 95 & 87 & 78 & 74 \\
fMA   & 93 & 92 & 78 & 68 & 81 \\
fMA2   & 94 & 95 & 93 & 93 & 92 \\
fARMA   & 93 & 94 & 93 & 95 & 91 \\
\bottomrule
\end{tabular}
\label{tab:gle}
\end{table}
As in Figure~\ref{fig:pos-est}, the LS estimator has the highest variance and fBM the largest bias.  In this case, however, the fSD and fMA estimators exhibit considerable bias in estimating $\alpha$, especially when $\tmin \gg \dt$.  In contrast, the fARMA estimator displays good accuracy and reasonable coverage even when $\tmin$ is $50\times$ the interobservation time $\dt$.

%-------------------------------------------------------------------------------

\section{Analysis of Experimental Data}\label{sec:exper}

We now investigate the performance of our high-frequency filters on a variety of real single-particle tracking experiments described in Table~\ref{tab:data}.
For each experiment, Table~\ref{tab:data} reports the interobservation time $\dt$, the number of particles $M$, the number of observations per trajectory $N$, and the type of camera and particle tracking software.  All tracked particles are inert polystyrene beads of diameter $d = \SI{1}{\micro\meter}$.
\begin{table}[!htb]
  \caption{Summary of experimental conditions for various single-particle tracking experiments.  The different types of fluids are water ($\water$), glycerol ($\gly$), mucus from human bronchial ephithelia cell cultures ($\HBE$), and polyethilene oxide ($\PEO$). The subscripts correspond to sampling frequency for $\water$, percent concentration for $\gly$, and percent weight (wt\%) for $\HBE$ and $\PEO$.  The two types of cameras are Flea3 USB 3.0~\citep[Flea3:][]{flea3} and Panoptes~\citep[Pan:][]{panoptes}.  The particle tracking software employed is either Video Spot Tracker~\citep[VS:][]{vstracker} or Net Tracker~\citep[Net:][]{newby.etal18}.}
  \centering
  \begin{tabular}{@{}cccccccc@{}}
    \toprule
    Medium &  Name & $D$ & $\dt$ (s) & $N$ & $M$ & Camera & Software \\ 
    \midrule
    \multirow{6}{*}{Viscous} & $\water_{15}$    & 0.43 & 1/15 & 1800 & 1293 & Flea3 &Net \\ 
    \multirow{6}{*}{($\alpha=1$)}& $\water_{30}$    & 0.43 & 1/30 & 1800 & 889 & Flea3 &Net \\ 
           & $\water_{60}$    & 0.43 & 1/60 &1800 & 1931 & Flea3 &Net\\ 
           & $\water_{60b}$    & 0.43 & 1/60 &1800 & 313 &Flea3 &VS \\ 
           & $\gly_{60}$    &0.09 & 1/60 &1800 & 532 & Flea3 &VS\\ 
           & $\gly_{80}$    &0.022 & 1/60 &1800 & 358 &Flea3 &VS\\
    \midrule
    \multirow{13}{*}{Viscoelastic} & $\HBE_{1.5}$    & - & 1/60 & 1800 & 63  &Flea3 &VS \\ 
    \multirow{13}{*}{($\alpha$ unknown)}& $\HBE_{2}$    & - & 1/60 & 1800 & 72  &Flea3 &VS\\ 
           & $\HBE_{2.5}$    & - & 1/60 &1800 & 76  &Flea3 &VS\\ 
           & $\HBE_{3}$    & - & 1/60 &1800 & 99  &Flea3 &VS\\ 
           & $\HBE_{4}$    & - & 1/60 &1800 & 180  &Flea3 &VS\\ 
           & $\HBE_{5}$    & - & 1/60 &1800 & 178  &Flea3 &VS\\ 
           & $\PEO_{0.22}$    & - & 1/38.17 &1145 & 123   &Pan &VS\\ 
           %& $\PEO_{0.3}$    &- & 1/38.17 &1145 & 176   &Pan &VS\\ 
           & $\PEO_{0.45}$    &- & 1/38.17 &1145 & 205  &Pan &VS \\
           & $\PEO_{0.6}$    & - & 1/38.17 &1145 & 192   &Pan &VS\\ 
           & $\PEO_{0.75}$    & - & 1/38.17 &1145 & 202   &Pan &VS\\ 
           & $\PEO_{0.9}$    & - & 1/38.17 &1145 & 124   &Pan &VS\\ 
           & $\PEO_{1.22}$    & - & 1/38.17 &1145 & 193   &Pan &VS\\ 
    \bottomrule
  \end{tabular}
  \label{tab:data}
\end{table}

\subsection{Viscous Fluids}\label{sec:viscous}

The first six experiments are conducted in viscous fluids (water and glycerol), for which $\alpha = 1$ and the diffusivity constant $D$ is derived from the Stokes-Einstein relation~\eqref{eq:stokes}.  %Figure~\ref{fig:vis-boxplot} and Table~\ref{tab:vis-cvg} display the
For the six estimators described in Section~\ref{sec:sim-loc}, estimates of $\aD$ and true coverage probabilities of the associated 95\% confidence intervals are displayed in Figure~\ref{fig:vis-boxplot} and Table~\ref{tab:vis-cvg}, respectively.
\begin{figure}[!htb]
  \centering
  \includegraphics[width = 1\textwidth]{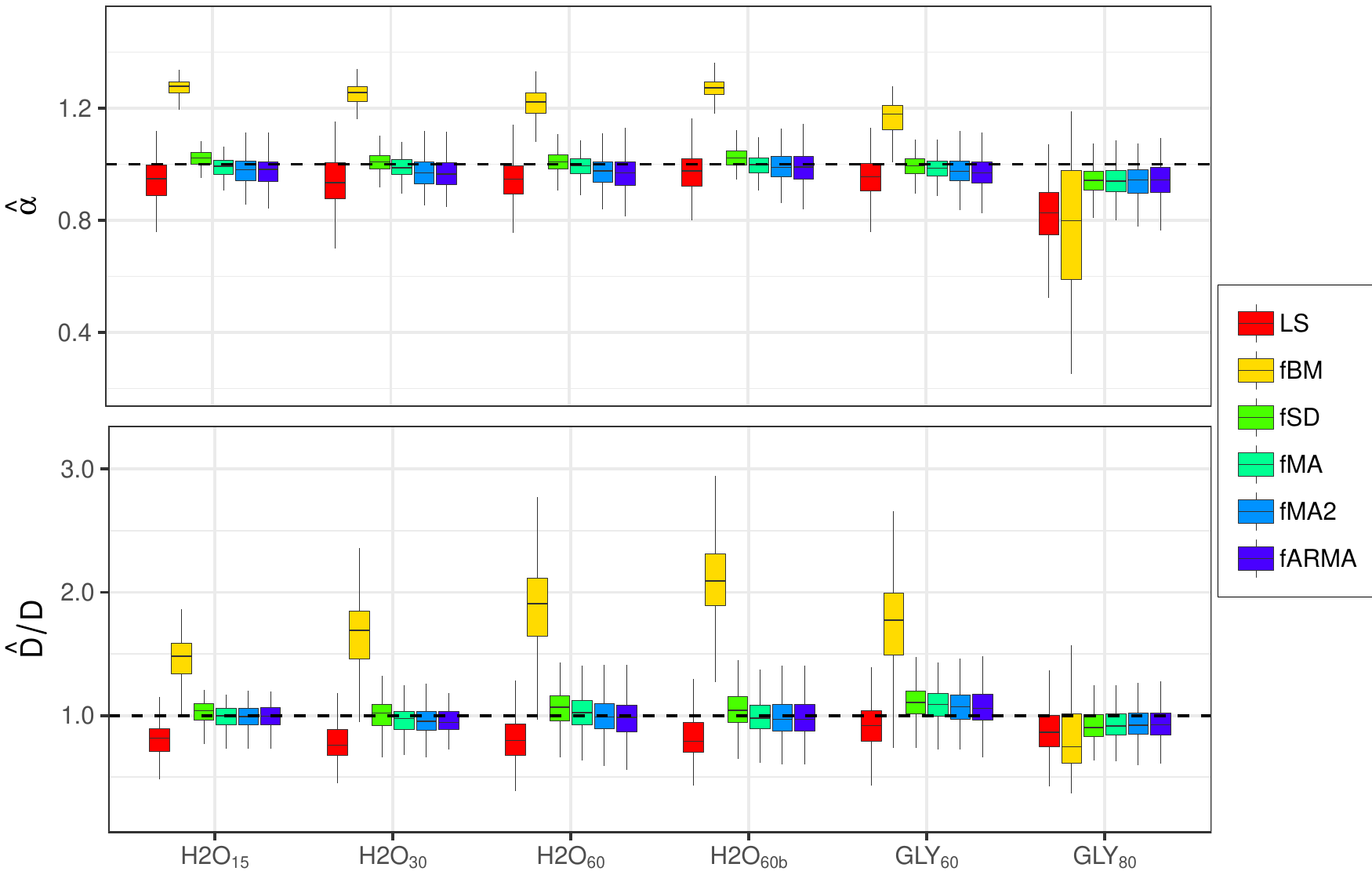}
  \caption{Estimates of $\aD$ for the viscous medium experiments in Table~\ref{tab:data}.}
  \label{fig:vis-boxplot}
\end{figure}
\begin{table}[!htb]
  \caption{Actual coverage by $95\%$ confidence intervals in viscous fluid study.}
  \centering
  \begin{tabular}{@{}lcccccc@{}}
    \toprule
    &  $\water_{15}$  &  $\water_{30}$  &  $\water_{60}$  &  $\water_{60b}$  &  $\gly_{60}$  &  $\gly_{80}$  \\ 
    \midrule
    fBM   &     0 &0 &0 &0 &4 &16 \\ 
    fSD   &     47 &42 &47 &11 &14 &44 \\ 
    fMA   &    94 &90 &93 &85 &90 &71 \\
    fMA2   &    95 &91 &92 &87 &91 &75 \\
    fARMA   &    95 &92 &94 &88 &92 &82 \\
    \bottomrule
  \end{tabular}
  \label{tab:vis-cvg}
\end{table}
\begin{table}[!htb]
  \caption{Ratio of true and estimated exposure time to interobservation time for the fSD model in the viscous medium experiments of Table~\ref{tab:data}.}
  \centering
  \begin{tabular}{@{}lcccccc@{}}
    \toprule
    &  $\water_{15}$  &  $\water_{30}$  &  $\water_{60}$ &  $\water_{60b}$  &  $\gly_{60}$  &  $\gly_{80}$ \\
    \midrule
    True $\tau / \dt$   &     $0.3$ & $0.3$ & $0.3$ & $0.3$ & $0.3$ & $0.3$ \\
    Estimated $\hat\tau / \dt$   &  0.93 & 0.91 & 0.89 & 0.91 & 0.85 & 0.54 \\
    \bottomrule
  \end{tabular}
  \label{tab:vis-tau}
\end{table}
Both fSD and the proposed high-frequency estimators remove most of the bias of fBM without camera error correction.  However, the fSD 95\% confidence intervals suffer from severe undercoverage, due to the parameter identifiability issue noted in Section~\ref{sec:sim-loc}.  Indeed, Table~\ref{tab:vis-tau} shows that the estimated exposure time $\hat \tau$ is
%several orders of magnitude \correct{not true for new 30\% exposure time, can be `` several times'' }
much larger than its true value $\tau$, as required in the $\water$ experiments to capture high-frequency MSD suppression. When $\tau$ is fixed at its true value, fSD estimation results are close those of fBM, %\correct{not that close, but still biased},
as illustrated in Figure~\ref{fig:control}.
                                                                           
\subsection{Viscoelastic Fluids}\label{sec:viscoel}

The remaining 12 experiments from Table~\ref{tab:data} are conducted in two kinds of viscoelastic media.  The first consists of mucus harvested from primary human bronchial epithelial (HBE) cell cultures \citep{hill.et.al14}. Washings from cultures were pooled and concentrated to desired weight percent solids (\wt).
    %     , including 1.5, 2 and 2.5 weight percent solids (wt\%).
Higher concentrations of solids in lung mucus have been associated with disease states, so an accurate recovery of biophysical properties is critical in samples with volumes too small to measure \wt{} directly~\citep{hill.et.al14}. The second medium, polyethylene oxide (\PEO), 
is a synthetic polyether compound with applications in diverse fields ranging from biomedicine to industrial manufacturing~\citep{working.et.al97}. The present data consists of trajectories in 5 megadalton (\si{\mega\dalton}) PEO at a range of \wt{} values. In all 12 viscoelastic experiments, subdiffusive motion $\alpha < 1$ is expected, but the true values of $\aD$ are unknown.

Figure~\ref{fig:viscoelastic} displays the various estimates of $\aD$ for the viscoelastic data.  The high-frequency noise models tend to produce similar results, with the largest differences occurring in the estimates of $\alpha$ at high \wt.
In the absence of %scientific theory to establish the
true values of $\aD$ against which to benchmark our models, 
% a baseline for comparison,
we compare the different subdiffusion estimators using the following metric.
\begin{figure}[htbp]
  \centering
  \includegraphics[width = 1\textwidth]{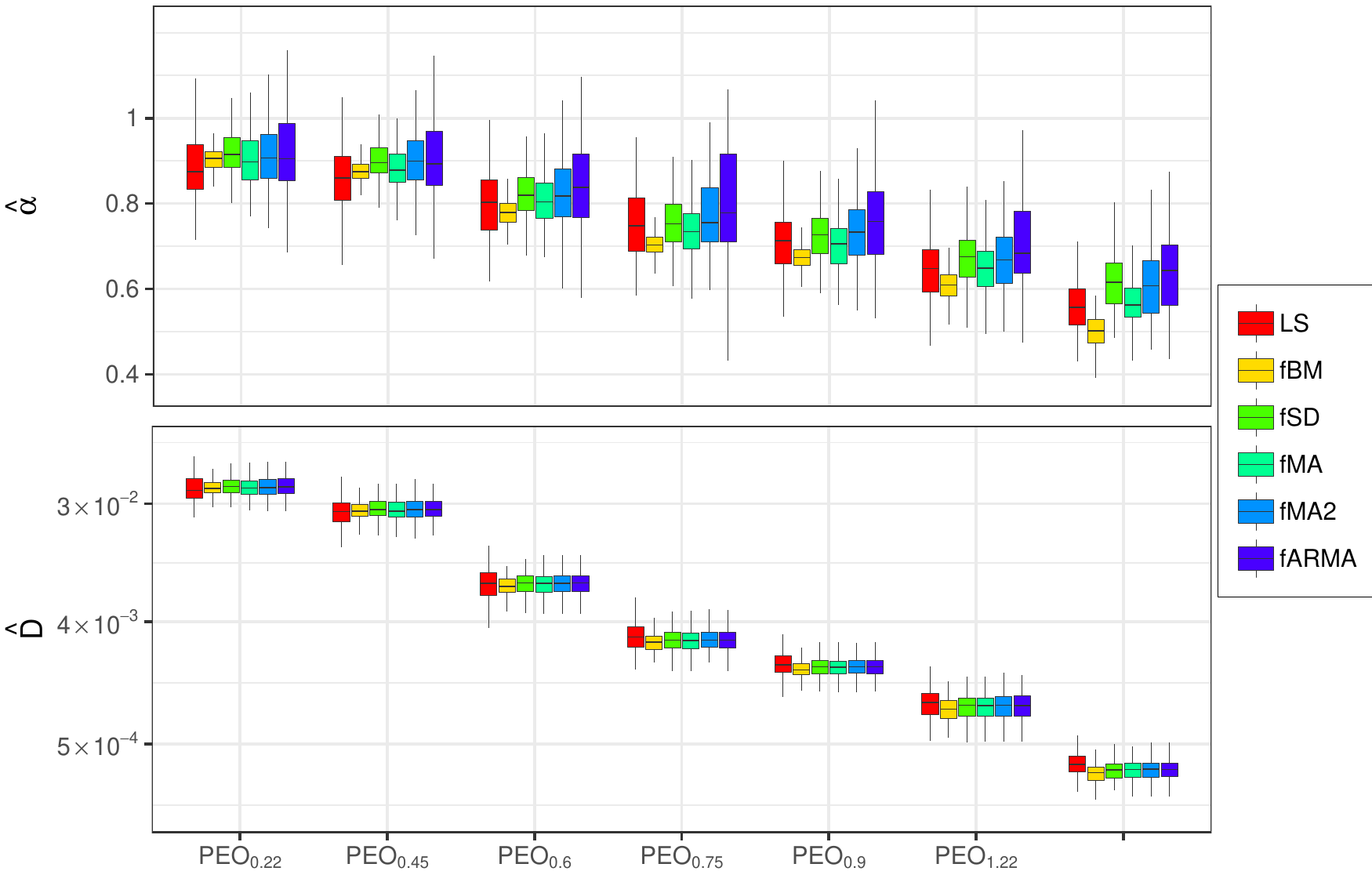}
  \includegraphics[width = 1\textwidth]{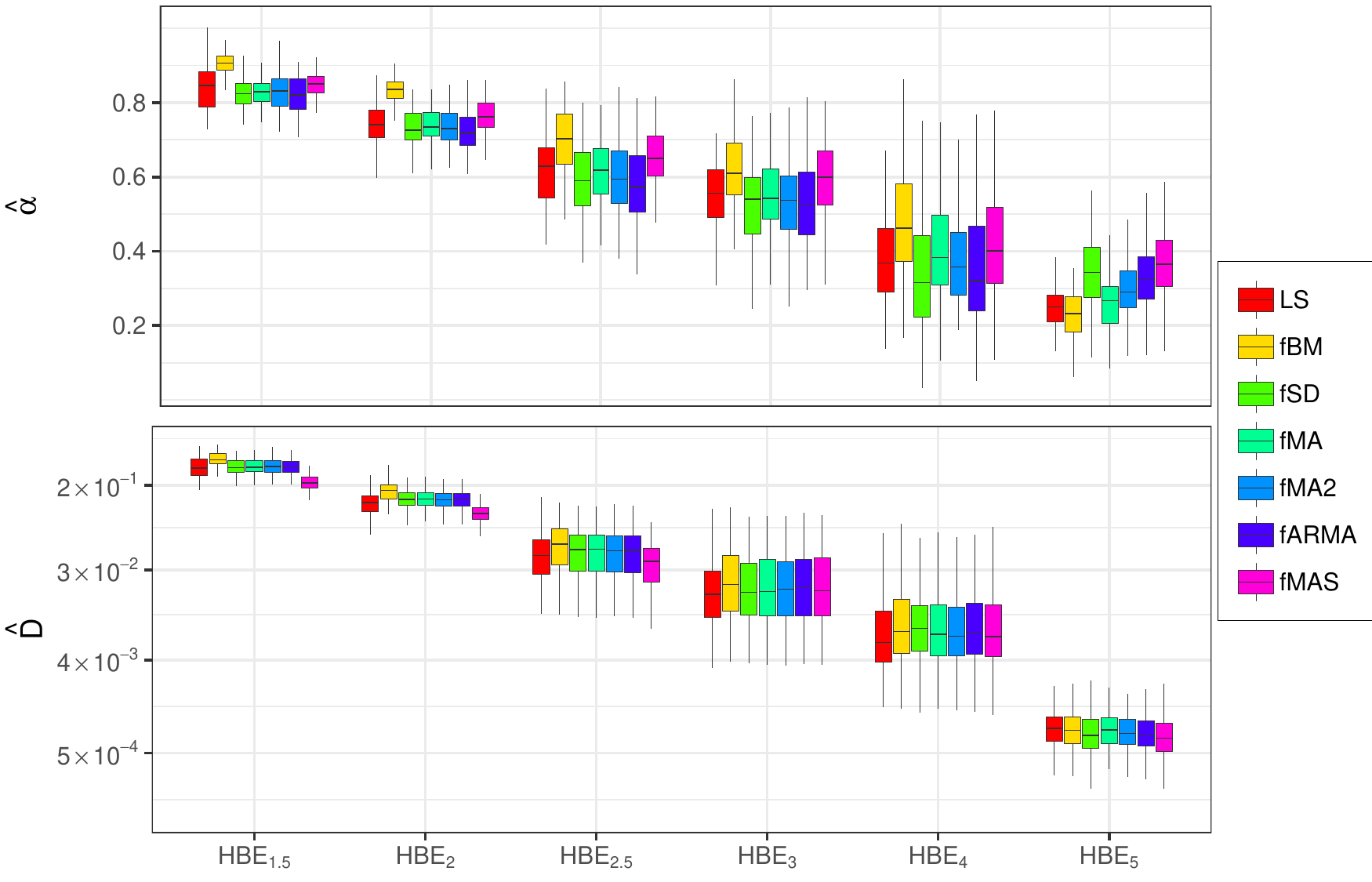}
  \caption{Estimates of $\aD$ for the viscoelastic medium experiments in Table~\ref{tab:data}.  For the $\HBE$ data, the subdiffusive estimators are the six described in Section~\ref{sec:sim-loc}, and that of the fMA + static noise (fMAS) model~\eqref{eq:fmas}.}
  \label{fig:viscoelastic}
\end{figure}

For measurements $\Y = (\Y_0, \Y_1, \ldots, \Y_N)$ of a given particle trajectory, let $\Y_{(r)k} = (\Y_{k},\Y_{k+r},\ldots,\Y_{k+ \lfloor N/r \rfloor r })$ denote the $k$th subset of the measurements downsampled by a factor of $r$.  Downsampling effectively removes all high-frequency dynamics from the particle positions, leading us initially to consider a subdiffusion estimator which maximizes the composite loglikelihood~\citep[e.g.,][]{varin.etal11}
\[
  \elc \sp r(\tth \mid \Y) = \sum_{k=0}^{r-1} \ell_{\textnormal{fBM}}(\tth \mid \Y_{(r)k}),
\]
where $\tth = (\alpha, \bbe, \SSi)$ are the parameters of the location-scale fBM model~\eqref{eq:lsmodel}.  However,
this estimator was found to have very high variance, which, for the purpose of constructing confidence intervals, was poorly estimated by the sandwich method~\citep{freedman06}.  Therefore, we have not pursued this downsampling estimator here.
Instead, we propose to evaluate the accuracy of subdiffusive model $M_j$ by calculating
% $\elc\sp r(\hat \tth\sp{M_j} \mid \Y)$, 
\begin{equation}\label{eq:dsc}
  %\dsc\sp r(\Y) =
  \elc\sp r(\hat \tth\sp{M_j} \mid \Y), %- \elc\sp r(\hat \tth\sp{\textnormal{fBM}} \mid \Y),
\end{equation}
where $\hat \tth \sp{M_j}$ are the corresponding elements of the MLE under $M_j$ for the complete set of measurements $\Y$.  Larger values of the composite likelihood statistic~\eqref{eq:dsc} 
% $\elc\sp r(\hat \tth\sp{M_j} \mid \Y)$ 
indicate better agreement with subdiffusive dynamics $\MSD_{\X}(t) = 2D \cdot t^\alpha$ for $t > \dt \times r$.  This approach to comparing models with respect to $\aD$ is evocative of the focused information criterion of~\cite{claeskens.hjort03}.

Table~\ref{tab:composite} reports the improvement in the composite likelihood statistic~\eqref{eq:dsc} of each measurement error model $M_j$ over the noise-free fBM model,
\[
  \dsc\sp r = \frac 1 M \sum_{m=1}^M \Big\{\elc\sp r(\hat \tth\sp{M_j} \mid \Y\sp m) - \elc\sp r(\hat \tth\sp{\textnormal{fBM}} \mid \Y\sp m)\Big\},
\]
where the average is calculated over the trajectories $\Y\sp 1, \ldots, \Y \sp M$ in each viscoelastic experiment of Table~\ref{tab:data}.  Interpretation of the units in Table~\ref{tab:composite} is similar to those of the AIC, upon multiplying ours by a factor of negative two.  However, we do not penalize by the number of parameters here, since all models have the same number of parameters in the subdiffusive range of interest.  We return to this point in the Discussion (Section~\ref{sec:disc}).
\begin{table}[!htb]
  \caption{Average improvement $\dsc\sp r$ in the composite likelihood statistic~\eqref{eq:dsc} relative to fBM for various subdiffusion estimators. For each experiment and downsampling factor $r$, the estimator with the greatest improvement is highlighted in bold.}\label{tab:composite}
  \centering
  \begin{tabular}{@{}lrccccccc@{}}
    \toprule
    \multicolumn{2}{c}{$\PEO$}
    && $0.22$ & $0.45$ & $0.6$ & $0.75$ & $0.9$ &$1.22$ \\ 
    \midrule
    \multirow{4}{*}{$r = 5$} 
    &fSD  && 3.1&2.9&4.2&4.3&3.6&7.6\\
    &fMA  && 2.9&2.5&3.7&4.3&3.8&11\\
    &fMA2  && 4.1&\textbf{4.6}&5.8&5.1&\textbf{4.8}&9.9\\
    &fARMA  && \textbf{4.8}&3.9&\textbf{6.9}&\textbf{5.2}&3.8&\textbf{12}\\
    \midrule
    \multirow{4}{*}{$r = 10$} 
    &fSD  && 2.2&2&2.9&3.5&2.9&5.7\\
    &fMA  &&  1.8&1.9&2.5&3.1&2.5&\textbf{8.7}\\
    &fMA2  &&  \textbf{2.7}&\textbf{3.4}&4.5&\textbf{3.6}&\textbf{3.6}&7.9\\
    &fARMA  &&  \textbf{2.7}&3&\textbf{4.7}&3.5&2.8&7.7\\
    \midrule
    \multirow{4}{*}{$r = 20$} 
    &fSD  &&  1.6&1.6&2.9&2.4&2.6&4.2\\
    &fMA   &&  \textbf{1.7}&1.6&2.3&2.4&1.9&\textbf{7}\\
    &fMA2   &&  1.5&\textbf{2.7}&\textbf{3.9}&\textbf{3.3}&\textbf{2.8}&6.1\\
    &fARMA   && 1.5&1.7&3.3&2&1.7&5\\
    \midrule
    \multicolumn{2}{c}{$\HBE$}
    && $1.5$ & $2$ & $2.5$ & $3$ & $4$ & $5$ \\ 
    \midrule
    \multirow{5}{*}{$r = 5$} 
    &fSD  && 15&29&\textbf{31}&28&29&-60\\
    &fMA  && 15&27&30&28&42&\textbf{0.06}\\
    &fMA2  && 15&\textbf{31}&\textbf{31}&\textbf{29}&\textbf{47}&-9.6\\
    &fARMA  && \textbf{16}&\textbf{31}&30&\textbf{29}&33&-22\\
    &fMAS  && 15&30&\textbf{31}&\textbf{29}&32&-72\\
    \midrule
    \multirow{5}{*}{$r = 10$} 
    &fSD  && 11&21&\textbf{23}&18&12&-53\\
    &fMA  &&  11&20&22&21&30&\textbf{0.25}\\
    &fMA2  &&  \textbf{12}&\textbf{22}&21&\textbf{22}&31&-7.1\\
    &fARMA  &&  11&\textbf{22}&22&20&18&-26\\
    &fMAS  && 11&21&\textbf{23}&19&13&-42\\
    \midrule
    \multirow{5}{*}{$r = 20$} 
    &fSD  &&  \textbf{9}&14&\textbf{16}&11&2.5&-61\\
    &fMA   &&  8.9&14&\textbf{16}&\textbf{18}&\textbf{23}&\textbf{0.81}\\
    &fMA2   &&  8.9&\textbf{17}&15&16&22&-5.3\\
    &fARMA   &&  8.1&16&14&11&7.1&-28\\
    &fMAS  && 9&14&15&16&11&-52\\
    \midrule
    \multirow{5}{*}{$r = 60$} 
    &fSD  &&  2.3&4.1&5.7&4.1&2.3&8\\
    &fMA   &&  2.1&4.3&\textbf{6.2}&\textbf{6.0}&8.5&1.3\\
    &fMA2   &&  2.3&\textbf{5.7}&5.1&5.3&\textbf{9.2}&5.3\\
    &fARMA   &&  \textbf{2.9}&5.1&5.4&4.1&2.7&4\\
    &fMAS  && 2.5&4.5&5.6&5.0&3.3&\textbf{12}\\
    \bottomrule
  \end{tabular}
\end{table}

As expected, noise correction produces significantly better estimates of $\aD$ than does the fBM model alone.  For the $\PEO$ data, the more accurate subdiffusion estimators are fMA2 and fARMA, whereas for $\HBE$ they are fMA and fMA2.  A notable exception is in the highest concentration $\HBE$ at 5~\wt, where for $r = 5,10,20$ all measurement error models except fMA are decisively dominated by noise-free fBM.  To see why this is the case, Figure~\ref{fig:11a} displays the empirical MSDs of three representative particle trajectories from the $\HBE$ 5~\wt{} dataset.  Each of these MSDs exhibits two distinct power-law signatures, with the changepoint occurring around $t = \SI{1}{\second}$.  Figure~\ref{fig:11b} displays the fitted MSD for various subdiffusion estimators.  We can see that fBM and fMA capture only the short-range power-law dynamics, whereas the other estimators capture the power law for $t > \SI{1}{\second}$.  However, for $r=5,10,20$, a sufficient amount of short-range power-law remains for it to outweigh the contribution of the longer-range dynamics in the calculation of the composite likelihood statistic~\eqref{eq:dsc}, thus favoring the fMB and fMA models.
\begin{figure}[htbp!]
  \centering
  \begin{subfigure}{\textwidth}
    \includegraphics[width = \textwidth]{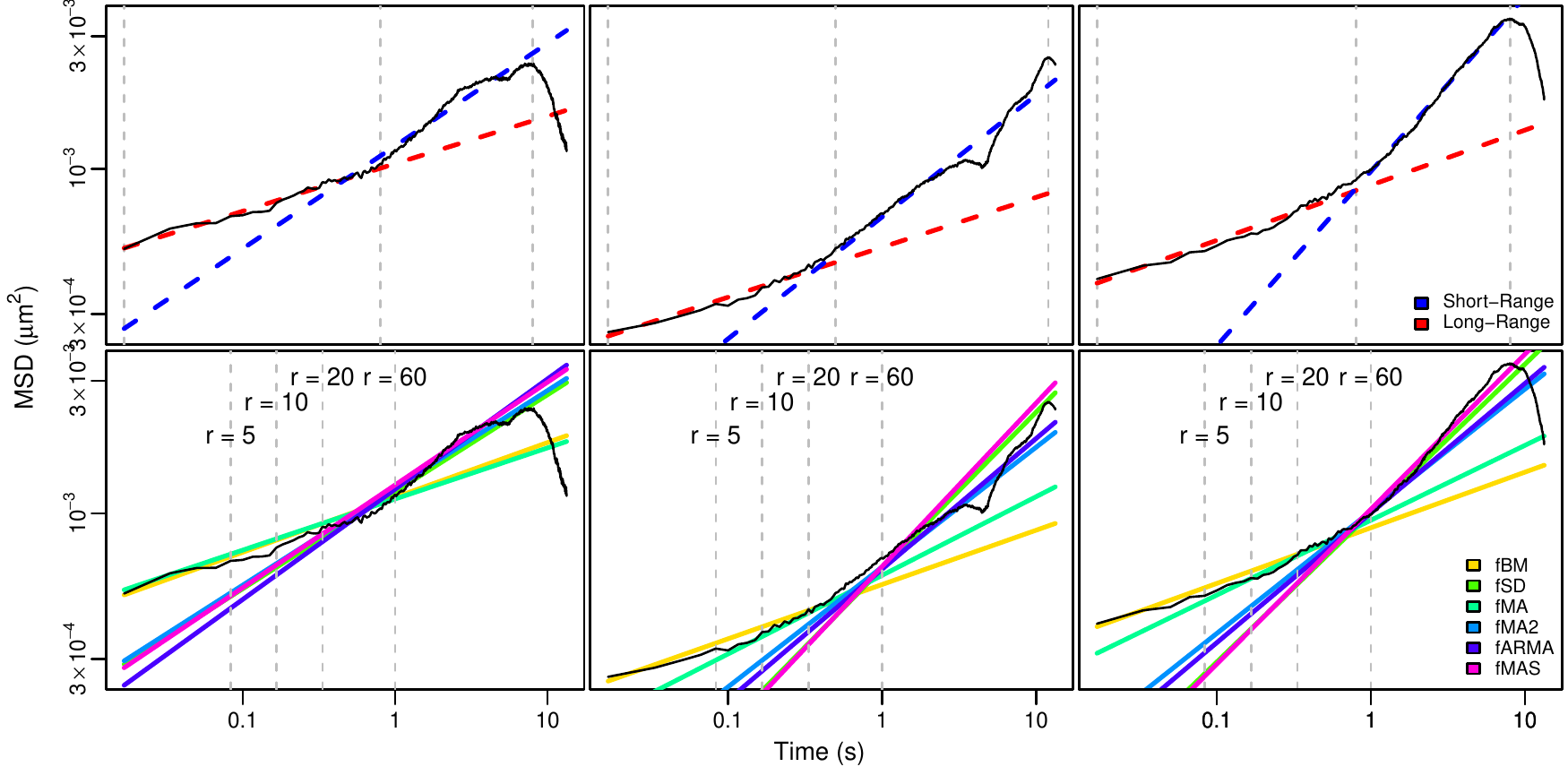}
    \phantomsubcaption\label{fig:11a}
    \phantomsubcaption\label{fig:11b}
  \end{subfigure}
  \caption{(a) Pathwise empirical MSD for 3 representative particles of diameter $\SI{1}{\micro\meter}$ with 5~\wt~mucus concentration, and their transient subdiffusion of two phases. The change point between two phases varies across particles.  (b) Empirical MSD and fitted subdiffusion with different methods, where the subdiffusion is computed using the power-law: $\MSD(t) = 2D\times t^\alpha$ and $\aD$ is extracted from  parametric estimations. Vertical dotted lines for different downsampling rates $r$ are also demonstrated.}\label{fig:11}
\end{figure}

It is theorized that the presence of two distinct power-law signatures in the $\HBE$ 5~\wt{} data is due to the extremely low particle mobility, such that the trajectory displacement signal is substantially masked by the measurement noise floor.  To investigate this, we added the static noise component of the Savin-Doyle model to the fMA model, leading to the so-called fMAS model
\begin{equation}\label{eq:fmas}
  \Y_n = (1-\rho) \X_n + \rho \X_{n-1} + \eps_n.
\end{equation}
Indeed, Table~\ref{tab:composite} indicates that fMAS most accurately captures long-range subdiffusion dynamics for $r = 60$.  It is noteworthy that fMAS outperforms the Savin-Doyle model (fSD) in this setting, suggesting that noise sources other than static and dynamic errors may be present in these data.

%-------------------------------------------------------------------------------

\section{Discussion}\label{sec:disc}

We present a family of parametric filters to correct for high-frequency noise in single-particle tracking measurements.  
% and demonstrate both theoretically and empirically its ability to account for a very broad range of localization errors.  
We demonstrate theoretically that our models can account for a very broad range of localization errors, and show how to combine them with arbitrary models of particle dynamics and low-frequency drift, so as to estimate subdiffusion parameters in a computationally efficient manner.

Compared to the state-of-the-art Savin-Doyle error model, our high-frequency filters generally exhibit lower bias, and much better coverage of confidence intervals for $\alpha \approx 1$, where the Savin-Doyle model suffers from a parameter identifiability issue.  A notable setting in which the Savin-Doyle model outperforms ours is when static noise dominates the high-frequency errors, e.g., in low-mobility experiments such as $\HBE$ 5~\wt.  Indeed, static noise is only covered by our definition of high-frequency noise~\eqref{hypo:msd} if the true position process $\X(t)$ is nonstationary (as is the case for fBM).  However, it is easy to combine static noise with our parametric filters without sacrificing computational efficiency, as we have done for the fMAS model in Section~\ref{sec:viscoel}.

An important practical question is how to determine which high-frequency error model produces the most accurate subdiffusion estimator for a given viscoelastic fluid and instrumental setup.
We have proposed a composite likelihood metric to approach this problem, but accounting for model complexity in the underlying estimation of Kullback-Liebler divergence would benefit from deeper theoretical and empirical investigation.  Possible directions of inquiry for the former are AIC for composite likelihoods~\citep{varin.etal11} and with consistent estimators~\citep{gronneberg.hjort14}, as well as focused information criteria for time series models~\citep{hermansen.etal15}.

%-------------------------------------------------------------------------------

\appendix

\section{Profile Likelihood for the Matrix-Normal Distribution}
\label{appendix:profile}

Let $\dX_{N\times k} = (\rv [0] \dX {N-1})$ denote the increments of the location-scale model~\eqref{eq:lsmodel} in matrix form.
Then $\dX$ follows a matrix-normal distribution~\eqref{eq:matnorm}
\begin{equation}\label{aeq:matnorm}
  \begin{aligned}
    \dX \sim \MN(\F\bbe, \VV_\pph, \SSi) \\
    \iff \qquad \vec(\dX) & \sim \N(\vec(\F\bbe), \SSi \otimes \VV_\pph),
  \end{aligned}
\end{equation}
where $\vec(\dX)$ concatenates the columns of $\dX$ into a vector of length $Nk$, similarly for $\vec(\F\bbe)$, and $\otimes$ denotes the Kronecker matrix product. 

As shown in~\cite{lysy.etal16}, the parameters $\tth = (\pph, \bbe, \SSi)$ of~\eqref{aeq:matnorm} can be efficiently estimated using a profile likelihood.  Consider a generalized matrix-normal model
\[
  \Y_{N\times k} \sim \MN(\F_\pph \bbe, \VV_\pph, \SSi),
\]
where both the design matrix $\F_\pph$ and the row-wise covariance $\VV_\pph$ depend on $\pph$.  Then for fix $\pph$, 
    %     That is, for fixed $\pph$
the conditional MLE of $(\bbe, \SSi)$ 
% $\argmax_{\bbe, \SSi} \ell(\tth \mid \dX)$
is given by
\begin{equation}
\begin{aligned}
&\hat{\bbe}_{\pph} = (\F_\pph'\VV_\pph^{-1}\F_\pph)^{-1}\F_\pph'\VV_\pph^{-1}\Y \\
&\hat{\SSi}_{\pph} = \frac{1}{N} (\Y - \F_\pph\hat{\bbe}_{\pph})' \VV_\pph^{-1} (\Y - \F_\pph\hat{\bbe}_{\pph}),
\end{aligned}
\end{equation}
from which we may calculate the profile loglikelihood
\begin{equation}
  \begin{aligned}
    \elp(\pph \mid \Y)
    & = \ell(\pph, \bbe = \hat \bbe_{\pph}, \SSi = \hat \SSi_{\pph} \mid \Y) \\
    & = -\tfrac{1}{2} \big\{k\log|\VV_\pph| + N\log|\hat \SSi_{\pph}| + Nk \big\}.
  \end{aligned}
\end{equation}
Upon solving the reduced optimization problem $\hat \pph = \argmax_{\pph} \elp(\pph \mid \Y)$,
% is the profile MLE of $\pph$, then $\hat \tth = (\hat \pph, \bbe_{\hat \pph}, \SSi_{\hat \pph})$ is the MLE of the full likelihood
we obtain $\hat \tth = (\hat \pph, \bbe_{\hat \pph}, \SSi_{\hat \pph})$ as the MLE of the full likelihood $\ell(\tth \mid \Y)$.  This technique can be used for all the measurement error models presented in this paper.

%-------------------------------------------------------------------------------

\section{Inference for the fSD Model}
\label{appendix:fdl-acf}

The $k$-dimensional fSD model~\eqref{eq:locmod} takes the form
\begin{equation}\label{aeq:fsd}
  \begin{aligned}
  \X(t) & = \sum_{j=1}^\np \bbe_j f_j(t) + \SSi^{1/2} \Z(t), \\
  \Y_n & = \frac 1 \tau \int_0^\tau \X(t_n - s) \ud s + \eps_n,
  \end{aligned}
\end{equation}
where $t_n = n\cdot\dt$, $\Z(t) = \big(Z_1(t), \ldots, Z_k(t)\big)$ with $Z_i(t) \iid B_\alpha(t)$, and $\eps_n \iid \N(\bm{0}, \sigma^2 \cdot \SSi)$ are independent of $\Z(t)$.  Letting $\dY_n = \Y_{n+1} - \Y_n$, we can rewrite~\eqref{aeq:fsd} to obtain
\[
  \dY_n = \sum_{j=1}^d \bbe_j \Delta \fs_{nj} + \SSi^{1/2} (\Delta \Zs_n - \Delta \eet_n),
\]
where
\[
  \fs_{nj} = \frac{1}{\tau}\int_{0}^\tau f_j(t_n-s) \ud s, \qquad Z_{ni}^\star = \frac 1 \tau \int_0^\tau Z_i(t_n-s) \ud s,
\]
and $\eet_n = \SSi^{-1/2}\eps_n \iid \N(\bm{0}, \sigma^2 \bm I_d)$.  Thus we have $\fs_{nj} = \frac 1 \tau \int_0^\tau f_j(t_n - s) \ud s$, 
% $\F_n = (\F_{n1}, \ldots, \F_{nd})$ is the $n$th row of $\F_{N\times d}$,
$\Zs_n = (Z_{n1}^\star, \ldots Z_{nk}^\star)$ with $Z_{ni}^\star = \frac 1 \tau \int_0^\tau Z_i(t_n-s) \ud s$, and $\eet_n = \SSi^{-1/2}\eps_n \iid \N(\bm{0}, \sigma^2 \bm I_d)$.  Thus, we have
\begin{equation}%\label{aeq:fsd}
  \dY_{N\times k} \sim \MN(\F \bbe, \VV_{\pph}, \SSi),
\end{equation}
where $\F_{N\times d}$ has elements $\F_{nj} = \Delta \fs_{nj}$, $\VV_{\pph}$ is a variance matrix parametrized by $\pph = (\alpha, \tau, \sigma)$ with elements
\begin{align*}
  \VV_{\pph}^{(n,m)}
  & = \cov(\Delta Z^\star_{ni} + \Delta \eta_{ni}, \Delta Z^\star_{mi} + \Delta \eta_{mi}) \\
  & = \cov(\Delta Z^\star_{ni}, \Delta Z^\star_{mi}) + \cov(\Delta \eta_{ni}, \Delta \eta_{mi}).
\end{align*}
To finish the calculations, without loss of generality we may focus on the one-dimensional case $Z_i(t) = Z(t) = B_\alpha(t)$ and $\eta_{in} = \eta_{n} \iid \N(0, \sigma^2)$.  Thus we have
\begin{align*}
  \cov(Z_n^\star, Z_m^\star)
  & = E[Z_n^\star Z_m^\star] \\
  & = \frac 1 {\tau^2} E\left[\int_{0}^{\tau} Z(t_n - s) \ud s \cdot \int_{0}^{\tau} Z(t_m - u) \ud u\right] \\
  & = \frac 1 {\tau^2} E\left[\int_0^\tau \int_0^\tau Z(t_n - s) Z(t_m - u) \ud s \ud u\right]\\
  & = \frac 1 {\tau^2} \int_0^\tau \int_0^\tau E\left[Z(t_n - s) Z(t_m - u) \right] \ud s \ud u,
\end{align*}
where the last line is obtained from the Fubini-Tonelli theorem, since by Cauchy-Schwarz we have
\begin{align*}
  \int_{0}^\tau\int_0^\tau E\Big[|Z(t_n - s) Z(t_m - u)| \Big]\ud s \ud u
  & \le \sqrt{\int_0^\tau E[Z(t_n - s)^2] \ud s \cdot \int_0^\tau E[Z(t_m - u)^2]\ud u} \\
  & = \sqrt{\int_0^\tau \MSD_Z(t_n - s) \ud s \cdot \int_0^\tau \MSD_Z(t_m - u)\ud u},
\end{align*}
and the right-hand side is finite as long as $\MSD_Z(t)$ is continuous for $t \ge 0$.  Thus, for the fBM process $Z(t) = B_\alpha(t)$ we have
\begin{align*}
  \cov(Z^\star_n, Z^\star_m)
  & = \frac 1 {2\tau^2} \int_0^\tau \int_0^\tau (t_n-s)^\alpha + (t_m-u)^\alpha - |(t_n-t_m) - (s-u)|^\alpha \ud s \ud u \\
  & = h_\tau(t_n) + h_\tau(t_m) - g_\tau(t_n - t_m),
\end{align*}
where
\[
  g_\tau(t) = \frac{ |t+\tau|^{\alpha+2} + |t-\tau|^{\alpha+2} - 2|t|^{\alpha+2}  }{2\tau^2(\alpha+1)(\alpha+2)}, \qquad h_\tau(t) = \frac{(t-\tau)^\alpha - t^\alpha}{2\tau(\alpha+1)}.
\]
Finally, since for any increment process $\Delta X_n$ we have
\begin{equation}\label{eq:covxy}
\cov(\Delta X_n, \Delta X_m) = E\left[ X_{n+1} X_{m+1} \right] - E\left[ X_{n+1}  X_m \right] - E\left[ X_n X_{m+1} \right] + E\left[ X_n X_m \right],
\end{equation}
we may calculate that
\begin{equation}\label{eq:fdyn_acf}
\acf_{\Delta Z^\star}(n) = \cov(\Delta Z^\star_n, \Delta Z^\star_{m+n}) = g_\tau(|n+1|\dt) + g_\tau(|n-1|\dt) - 2g_\tau(|n|\dt).
\end{equation}
Similarly, we obtain
\[
  \acf_{\Delta \eta}(n) = \sigma^2 \times \big\{2\cdot \id(n = 0) - \id(n=1)\big\},
\]
such that $\VV_{\pph}$ is a Toeplitz matrix with elements
\[
  \VV_{\pph}^{(n,m)} = \acf_{\Delta Z^\star}(n-m) + \acf_{\Delta \eta}(n-m).
\]

%-------------------------------------------------------------------------------

\section{Calculations for ARMA Noise Models}

\subsection{Relationship Between ACF and MSD}\label{appendix:msd2acf}

Let $X(t)$ be a one-dimensional CSI process with evenly-spaced observations $X_n = X(n\dt)$, such that
\[
  \MSD_X(n) = E[(X_n - X_0)^2].
\]
If $\Delta X_n = X_{n+1} - X_n$ is the corresponding increment process, then we have
\begin{equation}
\acf_{\Delta X}(n) = E[X_{n+1} X_1] + E[X_n X_0] - E[X_{n+1} X_0] - E[X_n X_1].
\end{equation}
Combined with the fact that
\[
  \MSD_{X}(n) = E[ X_n^2 ] + E[ X_0^2 ] - 2 E[ X_n X_0 ],
\]
we find that
\[
\acf_{\Delta X}(n) = \frac{1}{2} \{\MSD_{X}(|n-1|) + \MSD_{X}(|n+1|) - 2 \MSD_{X}(|n|)  \}.
\]
Conversely, we have
\[
  \MSD_{X}(n) = \MSD_{X}(n-1) + \acf_{\Delta X}(0) + 2 \sum_{h=1}^{n-1} \acf_{\Delta X}(h),
\]
such that
\begin{equation}\label{eq:acf2msd}
\MSD_{X}(n) = (n+1) \acf_{\Delta X}(0) + 2 \sum_{h=1}^{n} (n+1-h) \acf_{\Delta X}(h).
\end{equation}

\subsection{Autocorrelation Function of the $\armapq$ Filter}\label{appendix:farma-acf}

Consider a one-dimensional stationary increments process determined by the $\armapq$ filter~\eqref{eq:hqdrift},
\begin{equation}\label{eq:1d-farma}
\ddY_n = \sum_{i=1}^p \theta_i \ddY_{n-i} + \sum_{j=0}^q \rho_j \ddX_{n-j},
\end{equation}
for which the driving process $\ddX_n$ is assumed to have mean zero.  In the following subsections we shall calculate the autocorrelation function $\acf_{\ddY}(n)$ as a function of $\acf_{\ddX}(n) = \cov(\Delta X_m, \Delta X_{m+n})$.

\subsubsection{Autocorrelation of the  $\ma(q)$ Filter}\label{appendix:fma-acf}

For a purely moving-average process
\begin{equation}
\ddY_n = \sum_{i=0}^q \rho_i \ddX_{n-i},
\end{equation}
we have
\begin{equation}\label{eq:fma-acf}
  \acf_{\ddY}(n)
  % = E[\sum_{i=0}^q \rho_i \ddX_{t-i} \sum_{i=0}^q \rho_i \ddX_{t+n-i}]
  =  \sum_{i=0}^q \sum_{j=0}^q \rho_i\rho_j \acf_{\ddX} (n+i-j).
\end{equation}
This can be computed efficiently for all values of $\gga = (\rv [0] \gamma {N-1})$, $\gamma_n = \acf_{\ddY}(n)$,
% This can be computed efficiently for all values of $\gga = \big(\acf_{\ddY}(0), \ldots, \acf_{\ddY}(N-1)\big)$
using the following method.  Let $\eta_n = \acf_{\ddX}(n)$, $\bz_N$ denote the vector of $N$ zeros, and
for vectors $\bm a = (\rv a N)$ and $\bm b = (a_1, \rv [2] b M)$, let $\Toep(\bm a, \bm b)$ denote the $M \times N$ Toeplitz matrix with first row being $\bm a$ and first column $\bm b$:
\[
  \Toep(\bm a, \bm b) =
  \begin{bmatrix}
    a_1    & a_2    & a_3    & \cdots & \cdots & a_N    \\
    b_2    & a_1    & a_2    & \ddots &        & \vdots \\
    b_3    & b_2    & \ddots & \ddots & \ddots & \vdots \\
    \vdots & \ddots & \ddots & \ddots & a_2    & a_3    \\
    \vdots &        & \ddots & b_2    & a_1    & a_2    \\
    b_M    & \cdots & \cdots & b_3    & b_2    & a_1
  \end{bmatrix}.
\]
Then $\gga$ can be computed by the matrix multiplication
\[
  \gga = \Toep(\rrh_1, \rrh_2) \cdot \Toep(\eet_1, \eet_2) \cdot \rrh_0,
\]
where
\begin{align*}
  \eet_1 & = (\rv [0] \eta q), & \eet_2 & = (\rv [0] \eta {N+q}), \\
  \rrh_0 & = (\rv [0] \rho q), & \rrh_1 & = (\rrh_0, \bz_{N+1}), & \rrh_2 & = (\rho_0, \bz_{N-1}).
\end{align*}
Moreover, Toeplitz matrix-vector multiplication can be computed efficiently using the fast Fourier transform (FFT)~\citep[e.g.,][]{kailath.sayed99}. 
% This is done by a standard method of circulant embedding~\citep[e.g.,][]{kailath.sayed99} which we do not repeat here.  
That is, let $\FFT$ denote FFT the matrix of the appropriate dimension.  In order to compute $\gga$, we perform the following steps:
\begin{enumerate}
\item Let $\vv_3 = \FFT^{-1}(\FFT \vv_1 \odot \FFT \vv_2)$, where $\vv_1 = (\eet_2, 0, \rv [q] \eta 1)$, $\vv_2 = (\rrh_0, \bz_{N+q+1})$, and $\odot$ denotes the elementwise product between vectors.
\item Let $\vv_4$ denote the first $N+q+1$ elements of $\vv_3$.
\item Let $\vv_7 = \FFT^{-1}(\FFT \vv_5 \odot \FFT \vv_6)$, where $\vv_5 = (\rho_0, \bz_{2N}, \rv [q] \rho 1)$ and $\vv_6 = (\vv_4, \bz_{N})$.
\item $\gga$ is given by the first $N$ elements of $\vv_7$.
\end{enumerate}

\subsubsection{Autocorrelation of the $\ar(p)$ Filter}\label{appendix:far_acf}

For a purely autoregressive process
\begin{equation}
\ddY_n = \sum_{i=1}^{p} \theta_i \ddY_{n-i} + \ddX_n,
\end{equation}
the autocorrelation $\acf_{\ddY}(n)$ involves an infinite summation which generally cannot be simplified further.  Instead, we approximate the $\ar(p)$ filter with an $\ma(q)$ filter and use the result of Section~\ref{appendix:fma-acf}.  To do this, we rewrite $\ddY_n$ in terms of the lag operator $B$, such that
\[
  \ddY_n = \theta(B) \ddY_n + \ddX_n,
\]
where $\theta(x) = \theta_1 x + \cdots + \theta_p x^p$, and $B^k \ddY_n = \ddY_{n-k}$.  Rearranging terms and expanding into a power series, we find that
\begin{align*}
  % [1-\theta(B)]\ddY_n
      %       & = \ddX_n \\
  \ddY_n
              & \textstyle = [1-\theta(B)]^{-1}\ddX_n \\
              & \textstyle = \left[1 + \sum_{i=1}^\infty[\theta(B)]^i\right]\ddX_n = \left[1 + \sum_{i=1}^\infty\rho_i B^i\right] \ddX_n,
\end{align*}
such that $\ddY_n$ may be expressed as an $\ma(\infty)$ series.  Truncating to order $q$, the true autocorrelation $\acf_{\ddY}(n)$ is approximated by the autocorrelation~\eqref{eq:fma-acf} of the corresponding $\ma(q)$ process $\ddY_n \approx \sum_{i=1}^q \rho_i \ddX_{n-i}$.  The following lemma can be used to efficiently calculate the coefficients $\rho_i$.
\begin{lemma}\label{lemma1}
  Consider a polynomial $g(x) = \sum_{k=0}^p a_k x^k$ and its $n$-th power, $G(x) = [g(x)]^n = \sum_{k=0}^m b^{(n)}_k x^k$, where $m = n\cdot p$.   Then we have
  \begin{equation}
    \left[\tfrac{\ud}{\ud x} G(x)\right] g(x) = n \left[\tfrac{\ud}{\ud x} g(x)\right] G(x).
  \end{equation}
  As a result, when $a_0 \neq 0$ we can derive the coefficients of $G(x)$ recursively, with $b_0^{(n)} = a_0^n$ and
  \begin{equation}\label{aeq:polyrec}
    b_k^{(n)} = \frac{1}{k a_0} \times \left[ nk b_0^{(n)}a_k +  \sum_{i=1}^{k-1}(k-i)(n b_i^{(n)} a_{k-i} - a_i b_{k-i}^{(n)})\right].
  \end{equation}
\end{lemma}
Using Lemma~\ref{lemma1} with $g(x) = \theta(x)/x = \theta_1 + \cdots + \theta_p x^{p-1}$, we find that $\rho_i = \sum_{j=1}^i b_{i-j}^{(j)}$, 
where $b_{i-j}^{(j)}$ is given by~\eqref{aeq:polyrec} for $i-j \le j\cdot p$, and $b_{i-j}^{(j)} = 0$ otherwise.  In the simulations and data analyses of  sections~\ref{sec:sim} and~\ref{sec:exper}, we approximate all $\ar(p)$ filters by $\ma(50)$ filters.  Numerical experiments indicate that changing the order to $\ma(500)$ does not change the approximated autocorrelations by more than $10^{-14}$.

\subsubsection{Autocorrelation of the $\armapq$ Filter}

For the general $\armapq$ filter, we obtain the autocorrelation in two steps:
\begin{enumerate}
\item Let $\Delta Z_n = \sum_{j=0}^q \rho_j \ddX_{n-j}$, and calculate the autocorrelation of this $\ma(q)$ process using~\eqref{eq:fma-acf}.
\item Now we rewrite the original $\armapq$ process as
  \[
    \ddY_n = \sum_{i=1}^{p} \theta_i \ddY_{n-i} + \Delta Z_n,
  \]
  and we may approximate the autocorrelation of this $\ar(p)$ process by applying the technique of Appendix~\ref{appendix:far_acf} to $\acf_{\Delta Z}(n)$ obtained in Step 1.
\end{enumerate}

\subsection{Proof of Theorem~\ref{thm:restriction}}\label{appendix:thm1}

In order to parametrize the $\armapq$ filter such that it satisfies the high-frequency error hypothesis~\eqref{hypo:msd}, we begin by studying the relation between the MSD of a discrete-time univariate CSI process $\{X_n: n \ge 0\}$, and the power spectral density (PSD)
% $S_{\ddX}(\om)$
of its stationary increment process, $\ddX_n = X_{n+1} - X_n$.

For a stationary time series $\{\ddX_n: n \in \mathbb{Z}\}$ which is purely non-deterministic in the sense of the Wold decomposition~\citep[e.g.,][]{brockwell.davis91}, the PSD $S_{\ddX}(\om)$ is defined as the unique nonnegative symmetric integrable function for which the autocorrelation of $\ddX_n$ is given by
\begin{equation}\label{eq:psd2acf}
  \acf_{\ddX}(n) = \int_{-\pi}^{\pi}e^{-in \om} S_{\ddX}(\om) \ud \om.
\end{equation}

In order to prove Theorem~\ref{thm:restriction} we begin by proving the following lemma:
\begin{lemma}\label{lemma2}
  For two CSI process $X$ and $Y$ with corresponding increment processes $\ddX$ and $\ddY$, if $S_{\ddY}(\om)$ is positive in a neighborhood of $\om = 0$,
  and the PSD ratio satisfies
  \[
    \lim_{\om \rightarrow 0} \frac{S_{\ddX}(\om)}{S_{\ddY}(\om)} = 1,
  \]
  then $X$ and $Y$ satisfy the high-frequency error definition~\eqref{hypo:msd}, namely
  \[
    \lim_{n \rightarrow \infty} \frac{\MSD_{X}(n)}{\MSD_{Y}(n)} = 1.
  \]
\end{lemma}

\begin{proof}
  Using~\eqref{eq:acf2msd} and~\eqref{eq:psd2acf} we can relate $\MSD_{X}(n)$ to $S_{\ddX}(\om)$, such that
  \[
    \MSD_{X}(n+1) - \MSD_{X}(n) = \int_{-\pi}^{\pi}  \sum_{j = -n}^{n} e^{-i j \om} S_{\ddX}(\om) \ud\om = 	\int_{-\pi}^{\pi} D_n(\om) S_{\ddX}(\om) \ud\om,
  \]
  where $D_n(\om) = \sum_{j = -n}^{n} e^{-i j \om}$ is the $n$-th order Dirichlet kernel. Thus we have
  \begin{equation}\label{eq:psd2msd}
    \MSD_{X}(n) = \int_{-\pi}^{\pi} \sum_{k=0}^{n-1} D_k(\om) S_{\ddX}(\om) \ud\om = n \int_{-\pi}^{\pi} F_n(\om) S_{\ddX}(\om) \ud\om,
  \end{equation}
  where $F_n(\om) = \tfrac{1}{n} \sum_{k=0}^{n-1} D_k(\om)$ is the $n$-th order Fej\'er kernel. Since $F_n(\om)$ is symmetric about 0, we may rewrite $\MSD_{X}(n)$ as a convolution integral
  \begin{align*}
    \MSD_{X}(n) = n 2\pi \times \dfrac{1}{2\pi} \int_{-\pi}^{\pi} S_{\ddX}(\om) F_n(-\om) \ud\om = n 2\pi \times \{S_{\ddX} * F_n\}(0).
  \end{align*}

  By the Fej\'er kernel's summability property, we have
  \begin{align*}
    \{S_{\ddX} * F_n\}(\om) \rightarrow S_{\ddX}(\om) \quad \textnormal{a.e.}, \\
    \{S_{\ddY} * F_n\}(\om) \rightarrow S_{\ddY}(\om) \quad \textnormal{a.e.}.
  \end{align*}
  Since $S_{\ddY}(\om) > 0$ in a  neighborhood of $\om = 0$, we may thus find $\varepsilon > 0$ such that both $\{S_{\ddY} * F_n\}(0) \rightarrow S_{\ddY}(0) > 0$ and $\{S_{\ddY} * F_n\}(\om_0) \rightarrow S_{\ddY}(\om_0) > 0$ for $|\om_0| < \varepsilon$. Given this, we can express the MSD ratio as
  \begin{align*}
    \frac{\MSD_{X}(n)}{\MSD_{Y}(n)}
    & = \frac{\{S_{\ddX} * F_n\}(0)}{\{S_{\ddY} * F_n\}(0)} \\
    & = \frac{\{S_{\ddX} * F_n\}(0)}{\{S_{\ddY} * F_n\}(0)} - \frac{\{S_{\ddX} * F_n\}(\om_0)}{\{S_{\ddY} * F_n\}(\om_0)} \\
    & \phantom{=} + \frac{\{S_{\ddX} * F_n\}(\om_0)}{\{S_{\ddY} * F_n\}(\om_0)} - \frac{S_{\ddX}(\om_0)}{S_{\ddY}(\om_0)} + \frac{S_{\ddX}(\om_0)}{S_{\ddY}(\om_0)}.
  \end{align*}
  Since $S_{\ddX}(\om)$ and $F_n(\om)$ are both integrable and $\int F_n(\om) d\om = 1$, the convolution $\{S_{\ddX} * F_n\}(\om)$ is a uniformly continuous function.  The same argument applies to $\{S_{\ddY} * F_n\}(\om)$. Since $f_n(\om) = \frac{\{S_{\ddX} * F_n\}(\om)}{\{S_{\ddY} * F_n\}(\om)}$ is a ratio between two continuous functions, it is also a continuous function, which means that we can find $\om_1 > 0$ such that for $|\om| < \om_1$ we have $|f_n(0) - f_n(\om)| < \frac{\varepsilon}{3}$.
  % $\forall \varepsilon > 0$, $\exists \om_1$ such that $\forall \om < \om_1$, $|f_n(0) - f_n(\om)| < \frac{\varepsilon}{3}$.
  Moreover, by Fej\'er summability we have
  \[
    f_n(\om) = \frac{\{S_{\ddX} * F_n\}(\om)}{\{S_{\ddY} * F_n\}(\om)} \rightarrow \frac{S_{\ddX}(\om)}{S_{\ddY}(\om)} = f(\om) \quad \textnormal{a.e.},
  \]
  such that we may find $N_1$ such that $|f_n(\om) - f(\om)| < \tfrac{\varepsilon}{3}$ uniformly in $\om$ for $n > N_1$.
  % Thus we have $\forall \varepsilon > 0$, $\forall \om$, $\exists N_1$ such that $\forall n > N_1$, $|f_n(\om) - f(\om)| < \frac{\varepsilon}{3}$. 
  Thus, if
  \[
    \lim_{\om \rightarrow 0} \frac{S_{\ddX}(\om)}{S_{\ddY}(\om)} = 1,
  \]
  we may find $\om_2 > 0$ such that $|f(\om) - 1| < \tfrac{\varepsilon}{3}$ for $|\om| < \om_2$, and thus for $n > N_1$ and any $\om$ such that $|\om| < \min\{\om_1,\om_2\}$, we have
  % i.e., $\forall \varepsilon > 0$, $\exists \om_2$ such that $\forall \om < \om_2, |f(\om) - 1| < \frac{\varepsilon}{3}$, we have that $\forall \varepsilon > 0$, $\exists N_1$ such that $\forall n > N_1$
  \begin{align*}
    |\frac{\MSD_{X}(n)}{\MSD_{Y}(n)} - 1|
    & \le |f_n(0) - f_n(\om)| + |f_n(\om) - f(\om)| + |f(\om) - 1| \\
    & \le \frac{\varepsilon}{3} + \frac{\varepsilon}{3} + \frac{\varepsilon}{3} = \varepsilon,
  \end{align*}
  such that
  % where $0 < \om < \min\{\om_1, \om_2\}$, i.e.,
  \[
    \lim_{n \rightarrow \infty} \frac{\MSD_{X}(n)}{\MSD_{Y}(n)} = 1.
  \]
\end{proof}

To complete the proof of Theorem~\ref{thm:restriction}, 
we apply Lemma~\ref{lemma2} to the CSI process $X_n$ and its $\armapq$ filter $Y_n$ 
% $Y_n = \sum_{i=1}^p \theta_i Y_{n-i} + \sum_{j=0}^q \rho_j X_{n-i}$
as defined by~\eqref{eq:farma}.  That is, for the increment processes $\ddX_n$ and $\ddY_n = \sum_{i=1}^p \theta_i \ddY_{n-i} + \sum_{j=0}^q \rho_j \ddX_{n-i}$,
\begin{equation}
  \lim_{\om \rightarrow 0} \frac{S_{\ddX}(\om)}{S_{\ddY}(\om)} = \lim_{\om \rightarrow 0} \frac{|1 - \sum_{k=1}^p \theta_k \cdot e^{-i k \om}|^2}{|\sum_{j=0}^q \rho_j e^{-ij\om}|^2} =  \left\vert\frac{1-\sum_{i=1}^p \theta_i}{\sum_{j=0}^q \rho_j}\right\vert^2.
\end{equation}
Thus by setting $\rho_0 = 1 - \sum_{i=1}^p \theta_i - \sum_{j=1}^q \rho_j$, we have
\[
  \lim_{\om\rightarrow 0} \frac{S_{\ddX}(\om)}{S_{\ddY}(\om)} = \left(\frac{1-\sum_{i=1}^p \theta_i}{\sum_{j=0}^q \rho_j}\right)^2 = 1 \quad \implies \quad \lim_{n\rightarrow\infty} \frac{\MSD_{X}(n)}{\MSD_{Y}(n)} = 1,
\]
which completes the proof of Theorem~\ref{thm:restriction}.

\subsection{Proof of Theorem~\ref{thm:hfapprox}}\label{appendix:thm2}

The complete statement of Theorem~\ref{thm:hfapprox} is as follows.

Let $\Xts = \{X_n: n \ge 0\}$ denote the true positions of a CSI process, for which $\Yts \{Y_n: n \ge 0\}$ is the measurement process satisfying the high-frequency error definition~\eqref{hypo:msd}.
For the corresponding increment processes $\Delta \Xts = \{\ddX_n: n \in \mathbb{Z}\}$ and $\Delta \Yts = \{\ddY_n: n \in \mathbb{Z}\}$, suppose the PSD ratio
\[
  g(\om) = \frac{ S_{\ddY}(\om) }{ S_{\ddX}(\om) }
\]
is continuous on the interval $\om \in [-\pi, \pi]$.
Then there exists an $\armapq$ noise model $\bm{\mathcal{Y}^\star} = \{Y^\star_n: n \ge 0\}$ satisfying~\eqref{eq:farma} such that for all $n \ge 0$ we have
\begin{equation}\label{aeq:armaeps}
  \left\vert\frac{ \MSD_{Y^\star}(n)}{\MSD_{Y}(n)  } -1 \right\vert < \epsilon.	
\end{equation}.

\begin{proof}
  In order to show that there exits an $\armapq$ process
  \[
    Y^\star_n = \sum_{i=1}^p \theta_i Y^\star_{n-i} + \sum_{j=0}^q \rho_j X_{n-j},
  \]
  satisfying~\eqref{aeq:armaeps}, 
  we use~\eqref{eq:psd2msd} to write
  \begin{equation}
    \begin{split}
      \left\vert \frac{ \MSD_{Y^\star}(n)  }{ \MSD_{Y}(n)  } -1 \right\vert
      & = \frac{|\MSD_{Y^\star}(n) - \MSD_{Y}(n)|}{ \MSD_{Y}(n)  } \\
      & \leq \frac{ \int_{-\pi}^\pi F_n(\om) \cdot |S_{\ddY\star}(\om) - S_{\ddY}(\om)| \ud \om  }{ \int_{-\pi}^\pi F_n(\om) S_{\ddY}(\om) \ud \om } \\
      & = \frac{ \int_{-\pi}^\pi |r(\om) - g(\om)| \cdot F_n(\om) S_{\ddX}(\om) \ud \om }{ \int_{-\pi}^\pi F_n(\om) S_{\ddY}(\om) \ud \om },
    \end{split}
  \end{equation}
  where $g(\om) = S_{\ddY}(\om)/S_{\ddX}(\om)$ and
  \[
    r(\om) = \frac{S_{\ddY^\star}(\om)}{S_{\ddX}(\om)} = \left\vert\frac{\sum_{j=0}^q \rho_j e^{-ij\om}}{1 - \sum_{k=1}^p \theta_k \cdot e^{-i k \om}}\right\vert^2.
  \]
  Because $g(\om)$ is a ratio of nonnegative symmetric functions, it is also nonnegative symmetric, and since it is continuous, it satisfies the definition of a continuous PSD.  
  Therefore, by Corollary 4.4.1 of~\cite{brockwell.davis91}, we can find a stationary $\ma(q)$ process
  \[
    Z_n = \sum_{j=0}^q \rho_j \eta_{n-j}, \qquad \eta_n \iid \N(0, 1)
  \]
  satisfying parameter restrictions~\eqref{eq:stat}, such that if $S_{Z}(\om) = |\sum_{j=0}^q \rho_j e^{-ij\om}|^2$ is the PSD of this process,
  \[
    |S_Z(\om) - g(\om)| < \varepsilon_0 \textnormal{ for } \om \in [-\pi, \pi].
  \]
  Therefore, let $\ddY_n^\star = \sum_{j=0}^q \rho_j \ddX_n$, such that $r(\om) = S_{\ddY^\star}(\om)/S_{\ddX}(\om) = S_Z(\om) = |\sum_{j=0}^q \rho_j e^{-ij\om}|^2$.
Then we have
  \begin{equation}
    \begin{split}
      \left\vert \frac{ \MSD_{Y^\star}(n)  }{ \MSD_{Y}(n)  } -1 \right\vert
      & \le \frac{ \int_{-\pi}^\pi |r(\om) - g(\om)| \cdot F_n(\om) S_{\ddX}(\om) \ud \om }{ \int_{-\pi}^\pi F_n(\om) S_{\ddY}(\om) \ud \om } \\
      & \leq \varepsilon_0 \cdot \frac{ \int_{-\pi}^\pi  F_n(\om) S_{\ddX}(\om) \ud \om }{ \int_{-\pi}^\pi F_n(\om) S_{\ddY}(\om) \ud \om }
      = \varepsilon_0 \cdot\frac{\MSD_{X}(n)}{\MSD_{Y}(n)}.
    \end{split}
  \end{equation}
  Since $\lim_{n\rightarrow\infty} \MSD_{X}(n)/\MSD_{Y}(n)$ exists,  %$\frac{\MSD_{X}(n)}{\MSD_{Y}(n)}$ is bounded, i.e.
  there exists $L > 0$ such that for every $n$ we have
  \[
    0 \leq \frac{\MSD_{X}(n)}{\MSD_{Y}(n)} \leq L.
  \]
  Thus by letting $\varepsilon_0 = \varepsilon/L$, for every $n$ we have
  \[
    \left\vert \frac{ \MSD_{Y^\star}(n)  }{ \MSD_{Y}(n)  } -1 \right\vert < \varepsilon.
  \]
\end{proof}

%-------------------------------------------------------------------------------

\section{Calculations for the GLE Process}\label{appendix:GLE}

For the GLE process $X(t)$ defined by~\eqref{eq:gle} with sum-of-exponentials memory kernel
\[
  \phi(t) = \frac{\nu}{K} \sum_{k=1}^{K} \exp(-|t| \alpha_k),
\]
\cite{mckinley.etal09} derive its MSD to be
\begin{equation}\label{aeq:glemsd}
  \MSD_X(t) = \frac{2\kB T}{\nu/K} \left(C_0^2 t + \sum_{j=1}^{K-1} \frac{C_j^2}{r_j}(1-e^{-r_jt})\right), 
\end{equation}
where $r_1, \ldots, r_{K-1}$ are the roots of $q(y) = \prod_{k=1}^{K}(y - \alpha_k)$, and
\begin{equation}
C_0 = \left( \sum_{k=1}^K \frac{1}{\alpha_k} \right)^{1/2}, \qquad C_j = \frac{1}{r_j} \times \frac{ \sqrt{\sum_{k=1}^{K}\frac{1}{(1 - r_j \alpha_k)^2} } }{ (\sum_{k=1}^{K}\frac{\alpha_k}{1 - r_j \alpha_k})^2 - \sum_{k=1}^{K}\frac{\alpha_k}{1 - r_j \alpha_k} }.
\end{equation}
For the particular case of the Rouse memory kernel
\[
  \phi(t) = \frac{\nu}{K} \sum_{k=1}^{K} \exp(-|t| / \tau_k), \quad \tau_k = \tau \cdot (K / k)^{\gamma},
\]
\cite{mckinley.etal09} show that for sufficiently large $K$, the MSD exhibits (anomalous) transient subdiffusion,
\[
  \MSD_X(t) = \begin{cases}
    2 \cdot \Deff \cdot t^{\aeff}  & \tmin < t < \tmax \\
    2 \cdot D_{\textnormal{min}} \cdot t &  t<\tmin \\
    2 \cdot D_{\textnormal{max}} \cdot t &  t>\tmax.
  \end{cases}
\]
This is illustrated in Figure~\ref{fig:gle-fit} with $K = 300$ and GLE parameters $\gamma = 1.67$, $\tau = 0.01$, $\nu = 1$.
% GLE parameters $\pph = (\gamma, \tau, \nu) = (1.67, 0.01, 1)$.
\begin{figure}[htbp!]
  \centering
  \includegraphics[width = .8\textwidth]{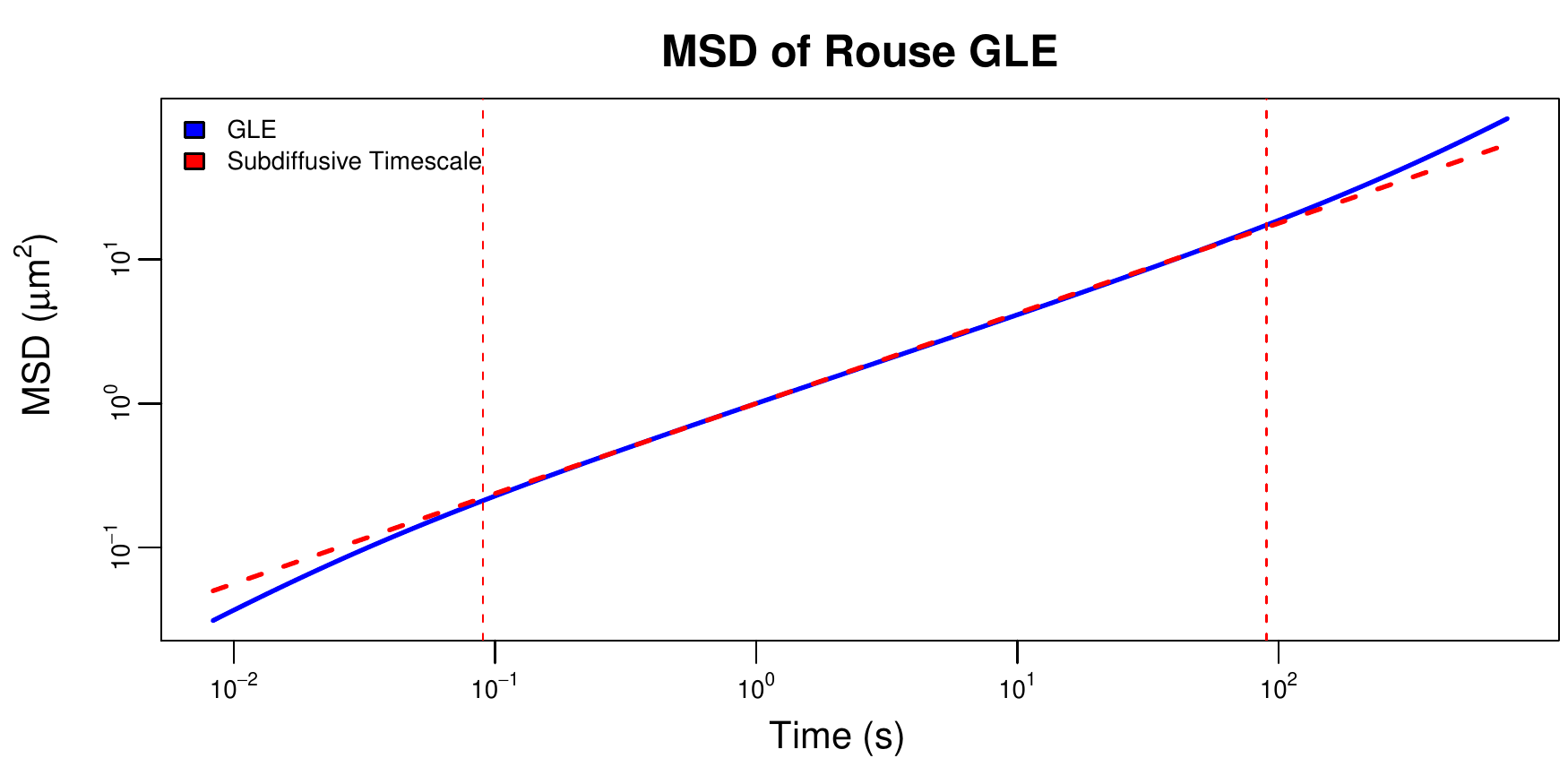}
  \caption{MSD of a Rouse GLE with $K = 300$ and $\gamma = 1.67$, $\tau = 0.01$, $\nu = 1$ (solid blue line).
    Also displayed is the subdiffusion timescale $(\tmin, \tmax)$ along with the power law $\MSD_X(t) = 2 \Deff \cdot t^{\aeff}$ on that range (red dotted lines).}
  \label{fig:gle-fit}
\end{figure}
Figure~\ref{fig:gle-fit} also displays the subdiffusion timescale $(\tmin, \tmax)$ along with the power law $\MSD_X(t) = 2 \Deff \cdot t^{\aeff}$ on that range.  The values of $(\tmin,\tmax,\aeff,\Deff)$ are determined from the GLE parameters $K$ and $\pph = (\gamma, \tau, \nu)$ via the following method.

\begin{enumerate}
\item Calculate $x_n = \log(t_n)$ and $y_n = \log \MSD_X(t_n \mid \pph, K)$ on a range of time points $t_0, \ldots, t_N$.  These should be picked on a fine grid such that $t_0 \ll \tmin$ and $t_N \gg \tmax$.
\item Let $\bt = (\tmin, \tmax)$, and let $\It = \{n: \tmin < t_n < \tmax\}$.  Then for any $\bt$ we calculate $\aeff^{(\bt)}$ and $\Deff^{(\bt)}$ via least-squares:
\begin{equation}
\aeff^{(\bt)} = \frac{\sum_{n\in \It}(y_n - \bar y)(x_n - \bar x)}{\sum_{n\in \It}(x_n - \bar x)^2}, \qquad \Deff^{(\bt)} = \tfrac 1 2 \exp(\bar y - \aeff^{(\bt)} \bar x),
\end{equation}
where $\bar x = \frac{1}{|\It|} \sum_{n\in\It} x_n$ and $\bar y = \frac{1}{|\It|} \sum_{n\in\It} y_n$ are the corresponding averages over the indices in $\It$.
\item The subdiffusion timescale $\bt$ is determined by solving the constrained optimization problem
\begin{multline*}
  \argmax_{\bt} |\log(\tmax) - \log(\tmin)| \\
  \textnormal{subject to} \qquad \max_{n\in \It} \left\vert\frac{\aeff^{(\bt)} \cdot x_n + \log(2\Deff^{(\bt)}) - y_n}{y_n}\right\vert < \kappa,
\end{multline*}
where $\kappa$ is a tolerance for departure from a perfect power law over the subdiffusive range.  In Figure~\ref{fig:gle-fit} and the calculations of Section~\ref{sec:gle} we have used $\kappa = 1\%$.  This optimization problem can be solved in $\mathcal O(N^2)$ steps by trying all combinations of $\tmin$ and $\tmax$ in the set $\{t_0, \ldots, t_N\}$.
\end{enumerate}

%-------------------------------------------------------------------------------

\bibliographystyle{stdref}
\bibliography{cmerr-ref}

\begin{thebibliography}{64}
\newcommand{\enquote}[1]{``#1''}
\providecommand{\natexlab}[1]{#1}
\providecommand{\url}[1]{\texttt{#1}}
\providecommand{\urlprefix}{URL }

\bibitem[{Amblard et~al.(1996)Amblard, Maggs, Yurke, Pargellis, and
  Leibler}]{amblard.etal96}
Amblard, F., Maggs, A.C., Yurke, B., Pargellis, A.N., and Leibler, S. (1996).
\newblock \enquote{Subdiffusion and anomalous local viscoelasticity in actin
  networks.}
\newblock \emph{Physical review letters}, 77(21): 4470.

\bibitem[{Ammar and Gragg(1988)}]{ammar.gragg88}
Ammar, G.S. and Gragg, W.B. (1988).
\newblock \enquote{Superfast solution of real positive definite toeplitz
  systems.}
\newblock \emph{SIAM Journal on Matrix Analysis and Applications}, 9(1):
  61--76.

\bibitem[{Ashley and Andersson(2015)}]{ashley.andersson15}
Ashley, T.T. and Andersson, S.B. (2015).
\newblock \enquote{Method for simultaneous localization and parameter
  estimation in particle tracking experiments.}
\newblock \emph{Physical Review E}, 92(5): 052707.

\bibitem[{Berglund(2010)}]{berglund10}
Berglund, A.J. (2010).
\newblock \enquote{Statistics of camera-based single-particle tracking.}
\newblock \emph{Physical Review E}, 82(1): 011917.

\bibitem[{Brockwell and Davis(1991)}]{brockwell.davis91}
Brockwell, P.J. and Davis, R.A. (1991).
\newblock \emph{Time Series: {Theory} and Methods}.
\newblock Springer-Verlag, New York.

\bibitem[{Bronstein et~al.(2009)Bronstein, Israel, Kepten, Mai, Shav-Tal,
  Barkai, and Garini}]{bronstein.etal09}
Bronstein, I., Israel, Y., Kepten, E., Mai, S., Shav-Tal, Y., Barkai, E., and
  Garini, Y. (2009).
\newblock \enquote{Transient anomalous diffusion of telomeres in the nucleus of
  mammalian cells.}
\newblock \emph{Physical review letters}, 103(1): 018102.

\bibitem[{Burov et~al.(2017)Burov, Figliozzi, Lin, Rice, Scherer, and
  Dinner}]{burov.etal17}
Burov, S., Figliozzi, P., Lin, B., Rice, S.A., Scherer, N.F., and Dinner, A.R.
  (2017).
\newblock \enquote{Single-pixel interior filling function approach for
  detecting and correcting errors in particle tracking.}
\newblock \emph{Proceedings of the National Academy of Sciences}, 114(2):
  221--226.

\bibitem[{Calderon(2016)}]{calderon16}
Calderon, C.P. (2016).
\newblock \enquote{Motion blur filtering: A statistical approach for extracting
  confinement forces and diffusivity from a single blurred trajectory.}
\newblock \emph{Physical Review E}, 93(5): 053303.

\bibitem[{Chenouard et~al.(2014)Chenouard, Smal, de~Chaumont, Ma{\v{s}}ka,
  Sbalzarini, Gong, Cardinale, Carthel, Coraluppi, Winter, Cohen, Godinez,
  Rohr, Kalaidzidis, Liang, Duncan, Shen, Xu, Magnusson, Jald\'en, Blau,
  Paul-Gilloteaux, Roudot, Kervrann, Waharte, Tinevez, Shorte, Willemse,
  Celler, van Wezel, Dan, Tsai, Ortiz~de Sol\'orzano, Olivo-Marin, and
  Meijering}]{chenouard.etal14}
Chenouard, N., Smal, I., de~Chaumont, F., Ma{\v{s}}ka, M., Sbalzarini, I.F.,
  Gong, Y., Cardinale, J., Carthel, C., Coraluppi, S., Winter, M., Cohen, A.R.,
  Godinez, W.J., Rohr, K., Kalaidzidis, Y., Liang, L., Duncan, J., Shen, H.,
  Xu, Y., Magnusson, K.E.G., Jald\'en, J., Blau, H.M., Paul-Gilloteaux, P.,
  Roudot, P., Kervrann, C., Waharte, F., Tinevez, J.Y., Shorte, S.L., Willemse,
  J., Celler, K., van Wezel, G.P., Dan, H.W., Tsai, Y.S., Ortiz~de Sol\'orzano,
  C., Olivo-Marin, J.C., and Meijering, E. (2014).
\newblock \enquote{Objective comparison of particle tracking methods.}
\newblock \emph{Nature methods}, 11(3): 281.

\bibitem[{{CISMM}(2019{\natexlab{a}})}]{panoptes}
{CISMM} (2019{\natexlab{a}}).
\newblock \enquote{Camera panoptes.}
\newblock
  \url{http://cismm.web.unc.edu/core-projects/force-microscopy/high-throughput-microscopy}.

\bibitem[{{CISMM}(2019{\natexlab{b}})}]{vstracker}
{CISMM} (2019{\natexlab{b}}).
\newblock \enquote{Video spot tracker.}
\newblock
  \url{http://cismm.web.unc.edu/resources/software-manuals/video-spot-tracker-manual}.

\bibitem[{Claeskens and Hjort(2003)}]{claeskens.hjort03}
Claeskens, G. and Hjort, N.L. (2003).
\newblock \enquote{The focused information criterion.}
\newblock \emph{Journal of the American Statistical Association}, 98(464):
  900--916.

\bibitem[{Deschout et~al.(2014)Deschout, Zanacchi, Mlodzianoski, Diaspro,
  Bewersdorf, Hess, and Braeckmans}]{deschout.etal14}
Deschout, H., Zanacchi, F.C., Mlodzianoski, M., Diaspro, A., Bewersdorf, J.,
  Hess, S.T., and Braeckmans, K. (2014).
\newblock \enquote{Precisely and accurately localizing single emitters in
  fluorescence microscopy.}
\newblock \emph{Nature methods}, 11(3): 253.

\bibitem[{Durbin(1960)}]{durbin60}
Durbin, J. (1960).
\newblock \enquote{The fitting of time series models.}
\newblock \emph{Review of the International Statistical Institute}, 28: 233 --
  243.

\bibitem[{Edward(1970)}]{edward70}
Edward, J.T. (1970).
\newblock \enquote{Molecular volumes and the stokes-einstein equation.}
\newblock \emph{Journal of Chemical Education}, 47(4): 261.

\bibitem[{Einstein(1956)}]{einstein56}
Einstein, A. (1956).
\newblock \emph{Investigations on the Theory of the Brownian Movement}.
\newblock Courier Corporation.

\bibitem[{Ernst et~al.(2017)Ernst, John, Guenther, Wagner, Schaefer, and
  Lehr}]{ernst.etal17}
Ernst, M., John, T., Guenther, M., Wagner, C., Schaefer, U.F., and Lehr, C.M.
  (2017).
\newblock \enquote{A model for the transient subdiffusive behavior of particles
  in mucus.}
\newblock \emph{Biophysical Journal}, 112(1): 172--179.

\bibitem[{Ferry(1980)}]{ferry80}
Ferry, J.D. (1980).
\newblock \emph{Viscoelastic properties of polymers}.
\newblock John Wiley \& Sons.

\bibitem[{{FLIR}(2019)}]{flea3}
{FLIR} (2019).
\newblock \enquote{Camera flear usb3.0.}
\newblock \url{https://www.ptgrey.com/products/flea3-usb3}.

\bibitem[{Fong et~al.(2013)Fong, Sharma, Fallica, Tierney, Fortune, and
  Zaman}]{fong.etal13}
Fong, E.J., Sharma, Y., Fallica, B., Tierney, D.B., Fortune, S.M., and Zaman,
  M.H. (2013).
\newblock \enquote{Decoupling directed and passive motion in dynamic systems:
  particle tracking microrheology of sputum.}
\newblock \emph{Annals of biomedical engineering}, 41(4): 837--846.

\bibitem[{Freedman(2006)}]{freedman06}
Freedman, D.A. (2006).
\newblock \enquote{On the so-called “huber sandwich estimator” and
  “robust standard errors”.}
\newblock \emph{The American Statistician}, 60(4): 299--302.

\bibitem[{Gal et~al.(2013)Gal, Lechtman-Goldstein, and Weihs}]{gal.etal13}
Gal, N., Lechtman-Goldstein, D., and Weihs, D. (2013).
\newblock \enquote{Particle tracking in living cells: A review of the mean
  square displacement method and beyond.}
\newblock \emph{Rheologica Acta}, 52(5): 425--443.

\bibitem[{Gr{\o}nneberg and Hjort(2014)}]{gronneberg.hjort14}
Gr{\o}nneberg, S. and Hjort, N.L. (2014).
\newblock \enquote{The copula information criteria.}
\newblock \emph{Scandinavian Journal of Statistics}, 41(2): 436--459.

\bibitem[{Hermansen et~al.(2015)Hermansen, Hjort, and
  Jullum}]{hermansen.etal15}
Hermansen, G.H., Hjort, N.L., and Jullum, M. (2015).
\newblock \enquote{Parametric or nonparametric: The fic approach for stationary
  time series.}
\newblock In F.J. Samaniego, editor, \emph{Proceedings of the 60th World
  Statistics Congress of the International Statistical Institute}, pages
  4827--4832. The International Statistical Institute.

\bibitem[{Hill et~al.(2014{\natexlab{a}})Hill, Vasquez, Mellnik, McKinley,
  Vose, Mu, Henderson, Donaldson, Alexis, Boucher et~al.}]{hill.etal14}
Hill, D.B., Vasquez, P.A., Mellnik, J., McKinley, S.A., Vose, A., Mu, F.,
  Henderson, A.G., Donaldson, S.H., Alexis, N.E., Boucher, R.C. et~al.
  (2014{\natexlab{a}}).
\newblock \enquote{A biophysical basis for mucus solids concentration as a
  candidate biomarker for airways disease.}
\newblock \emph{PloS one}, 9(2): e87681.

\bibitem[{Hill et~al.(2014{\natexlab{b}})Hill, Vasquez, Mellnik, McKinley,
  Vose, Mu, Henderson, Donaldson, Alexis, Boucher et~al.}]{hill.et.al14}
Hill, D.B., Vasquez, P.A., Mellnik, J., McKinley, S.A., Vose, A., Mu, F.,
  Henderson, A.G., Donaldson, S.H., Alexis, N.E., Boucher, R.C. et~al.
  (2014{\natexlab{b}}).
\newblock \enquote{A biophysical basis for mucus solids concentration as a
  candidate biomarker for airways disease.}
\newblock \emph{PloS one}, 9(2): e87681.

\bibitem[{Kailath et~al.(1979)Kailath, Kung, and Morf}]{kailath.etal79}
Kailath, T., Kung, S.Y., and Morf, M. (1979).
\newblock \enquote{Displacement ranks of matrices and linear equations.}
\newblock \emph{Journal of Mathematical Analysis and Applications}, 68(2):
  395--407.

\bibitem[{Kailath and Sayed(1999)}]{kailath.sayed99}
Kailath, T. and Sayed, A.H., editors (1999).
\newblock \emph{Fast Reliable Algorithms for Matrices with Structure}.
\newblock Society for Industrial and Applied Mathematics, Philadelphia.

\bibitem[{Koslover et~al.(2016)Koslover, Chan, and Theriot}]{koslover.etal16}
Koslover, E.F., Chan, C.K., and Theriot, J.A. (2016).
\newblock \enquote{Disentangling random motion and flow in a complex medium.}
\newblock \emph{Biophysical journal}, 110(3): 700--709.

\bibitem[{Kou(2008)}]{kou08}
Kou, S.C. (2008).
\newblock \enquote{Stochastic modeling in nanoscale biophysics: subdiffusion
  within proteins.}
\newblock \emph{The Annals of Applied Statistics}, 2(2): 501--535.

\bibitem[{Kowalczyk et~al.(2014)Kowalczyk, Oelschlaeger, and
  Willenbacher}]{kowalczyk.etal14}
Kowalczyk, A., Oelschlaeger, C., and Willenbacher, N. (2014).
\newblock \enquote{Tracking errors in 2d multiple particle tracking
  microrheology.}
\newblock \emph{Measurement Science and Technology}, 26(1): 015302.

\bibitem[{Kubo(1966)}]{kubo66}
Kubo, R. (1966).
\newblock \enquote{The fluctutation-dissipation theorem.}
\newblock \emph{Reports on Progress in Physics}, 29: 255--284.

\bibitem[{Lai et~al.(2007)Lai, O'Hanlon, Harrold, Man, Wang, Cone, and
  Hanes}]{lai.et.al07}
Lai, S.K., O'Hanlon, D.E., Harrold, S., Man, S.T., Wang, Y.Y., Cone, R., and
  Hanes, J. (2007).
\newblock \enquote{Rapid transport of large polymeric nanoparticles in fresh
  undiluted human mucus.}
\newblock \emph{Proceedings of the National Academy of Sciences}, 104(5):
  1482--1487.

\bibitem[{Lee et~al.(2007)Lee, Roichman, Yi, Kim, Yang, Van~Blaaderen,
  Van~Oostrum, and Grier}]{lee.etal07}
Lee, S.H., Roichman, Y., Yi, G.R., Kim, S.H., Yang, S.M., Van~Blaaderen, A.,
  Van~Oostrum, P., and Grier, D.G. (2007).
\newblock \enquote{Characterizing and tracking single colloidal particles with
  video holographic microscopy.}
\newblock \emph{Optics Express}, 15(26): 18275--18282.

\bibitem[{Levinson(1947)}]{levinson47}
Levinson, N. (1947).
\newblock \enquote{The {Wiener} {RMS} error criterion in filter design and
  prediction.}
\newblock \emph{Journal Of Mathematical Physics}, 25: 261 -- 278.

\bibitem[{Ling and Lysy(2017)}]{ling.lysy17}
Ling, Y. and Lysy, M. (2017).
\newblock \emph{SuperGauss: Superfast Likelihood Inference for Stationary
  Gaussian Time Series}.
\newblock \urlprefix\url{https://CRAN.R-project.org/package=SuperGauss}.
\newblock R package version 1.0.

\bibitem[{Lysy et~al.(2016)Lysy, Pillai, Hill, Forest, Mellnik, Vasquez, and
  McKinley}]{lysy.etal16}
Lysy, M., Pillai, N.S., Hill, D.B., Forest, M.G., Mellnik, J.W., Vasquez, P.A.,
  and McKinley, S.A. (2016).
\newblock \enquote{Model comparison and assessment for single particle tracking
  in biological fluids.}
\newblock \emph{Journal of the American Statistical Association}, 111(516):
  1413--1426.

\bibitem[{Mason et~al.(1997)Mason, Ganesan, Van~Zanten, Wirtz, and
  Kuo}]{mason.etal97}
Mason, T., Ganesan, K., Van~Zanten, J., Wirtz, D., and Kuo, S. (1997).
\newblock \enquote{Particle tracking microrheology of complex fluids.}
\newblock \emph{Physical Review Letters}, 79(17): 3282.

\bibitem[{Mason and Weitz(1995)}]{mason.weitz95}
Mason, T.G. and Weitz, D. (1995).
\newblock \enquote{Optical measurements of frequency-dependent linear
  viscoelastic moduli of complex fluids.}
\newblock \emph{Physical review letters}, 74(7): 1250.

\bibitem[{McKinley et~al.(2009)McKinley, Yao, and Forest}]{mckinley.etal09}
McKinley, S.A., Yao, L., and Forest, M.G. (2009).
\newblock \enquote{Transient anomalous diffusion of tracer particles in soft
  matter.}
\newblock \emph{Journal of Rheology}, 53(6): 1487--1506.

\bibitem[{Mellnik et~al.(2016)Mellnik, Lysy, Vasquez, Pillai, Hill, Cribb,
  McKinley, and Forest}]{mellnik.etal16}
Mellnik, J.W., Lysy, M., Vasquez, P.A., Pillai, N.S., Hill, D.B., Cribb, J.,
  McKinley, S.A., and Forest, M.G. (2016).
\newblock \enquote{Maximum likelihood estimation for single particle, passive
  microrheology data with drift.}
\newblock \emph{Journal of Rheology}, 60(3): 379--392.

\bibitem[{Michalet(2010)}]{michalet10}
Michalet, X. (2010).
\newblock \enquote{Mean square displacement analysis of single-particle
  trajectories with localization error: Brownian motion in an isotropic
  medium.}
\newblock \emph{Physical Review E}, 82(4): 041914.

\bibitem[{Michalet and Berglund(2012)}]{michalet.berglund12}
Michalet, X. and Berglund, A.J. (2012).
\newblock \enquote{Optimal diffusion coefficient estimation in single-particle
  tracking.}
\newblock \emph{Physical Review E}, 85(6): 061916.

\bibitem[{Mortensen et~al.(2010)Mortensen, Churchman, Spudich, and
  Flyvbjerg}]{mortensen.etal10}
Mortensen, K.I., Churchman, L.S., Spudich, J.A., and Flyvbjerg, H. (2010).
\newblock \enquote{Optimized localization analysis for single-molecule tracking
  and super-resolution microscopy.}
\newblock \emph{Nature methods}, 7(5): 377.

\bibitem[{Newby et~al.(2018)Newby, Schaefer, Lee, Forest, and
  Lai}]{newby.etal18}
Newby, J.M., Schaefer, A.M., Lee, P.T., Forest, M.G., and Lai, S.K. (2018).
\newblock \enquote{Convolutional neural networks automate detection for
  tracking of submicron-scale particles in 2d and 3d.}
\newblock \emph{Proceedings of the National Academy of Sciences}, 115(36):
  9026--9031.

\bibitem[{Qian et~al.(1991)Qian, Sheetz, and Elson}]{qian.etal91}
Qian, H., Sheetz, M.P., and Elson, E.L. (1991).
\newblock \enquote{Single particle tracking. analysis of diffusion and flow in
  two-dimensional systems.}
\newblock \emph{Biophysical journal}, 60(4): 910--921.

\bibitem[{Rowlands and So(2013)}]{rowlands.so13}
Rowlands, C.J. and So, P.T. (2013).
\newblock \enquote{On the correction of errors in some multiple particle
  tracking experiments.}
\newblock \emph{Applied physics letters}, 102(2): 021913.

\bibitem[{Savin and Doyle(2005)}]{savin.doyle05}
Savin, T. and Doyle, P.S. (2005).
\newblock \enquote{Static and dynamic errors in particle tracking
  microrheology.}
\newblock \emph{Biophysical journal}, 88(1): 623--638.

\bibitem[{Saxton and Jacobson(1997)}]{saxton.jacobson97}
Saxton, M.J. and Jacobson, K. (1997).
\newblock \enquote{Single-particle tracking: applications to membrane
  dynamics.}
\newblock \emph{Annual review of biophysics and biomolecular structure}, 26(1):
  373--399.

\bibitem[{Sikora et~al.(2017)Sikora, Teuerle, Wy{\l}oma{\'n}ska, and
  Grebenkov}]{sikora.etal17}
Sikora, G., Teuerle, M., Wy{\l}oma{\'n}ska, A., and Grebenkov, D. (2017).
\newblock \enquote{Statistical properties of the anomalous scaling exponent
  estimator based on time-averaged mean-square displacement.}
\newblock \emph{Physical Review E}, 96(2): 022132.

\bibitem[{Soussou et~al.(1970)Soussou, Moavenzadeh, and
  Gradowczyk}]{soussou.et.al70}
Soussou, J., Moavenzadeh, F., and Gradowczyk, M. (1970).
\newblock \enquote{Application of prony series to linear viscoelasticity.}
\newblock \emph{Transactions of the Society of Rheology}, 14(4): 573--584.

\bibitem[{Suh et~al.(2005)Suh, Dawson, and Hanes}]{suh.etal05}
Suh, J., Dawson, M., and Hanes, J. (2005).
\newblock \enquote{Real-time multiple-particle tracking: applications to drug
  and gene delivery.}
\newblock \emph{Advanced drug delivery reviews}, 57(1): 63--78.

\bibitem[{Szymanski and Weiss(2009)}]{szymanski.weiss09}
Szymanski, J. and Weiss, M. (2009).
\newblock \enquote{Elucidating the origin of anomalous diffusion in crowded
  fluids.}
\newblock \emph{Physical review letters}, 103(3): 038102.

\bibitem[{van~der Schaar et~al.(2008)van~der Schaar, Rust, Chen, van~der
  Ende-Metselaar, Wilschut, Zhuang, and Smit}]{vander.etal08}
van~der Schaar, H.M., Rust, M.J., Chen, C., van~der Ende-Metselaar, H.,
  Wilschut, J., Zhuang, X., and Smit, J.M. (2008).
\newblock \enquote{Dissecting the cell entry pathway of dengue virus by
  single-particle tracking in living cells.}
\newblock \emph{PLoS pathogens}, 4(12): e1000244.

\bibitem[{Varin et~al.(2011)Varin, Reid, and Firth}]{varin.etal11}
Varin, C., Reid, N., and Firth, D. (2011).
\newblock \enquote{An overview of composite likelihood methods.}
\newblock \emph{Statistica Sinica}, pages 5--42.

\bibitem[{Vestergaard et~al.(2014)Vestergaard, Blainey, and
  Flyvbjerg}]{vestergaard.etal14}
Vestergaard, C.L., Blainey, P.C., and Flyvbjerg, H. (2014).
\newblock \enquote{Optimal estimation of diffusion coefficients from
  single-particle trajectories.}
\newblock \emph{Physical Review E}, 89(2): 022726.

\bibitem[{Wang et~al.(2008)Wang, Lai, Suk, Pace, Cone, and
  Hanes}]{wang.et.al08}
Wang, Y.Y., Lai, S.K., Suk, J.S., Pace, A., Cone, R., and Hanes, J. (2008).
\newblock \enquote{Addressing the peg mucoadhesivity paradox to engineer
  nanoparticles that “slip” through the human mucus barrier.}
\newblock \emph{Angewandte Chemie International Edition}, 47(50): 9726--9729.

\bibitem[{Weihs et~al.(2007)Weihs, Teitell, and Mason}]{weihs.etal07}
Weihs, D., Teitell, M.A., and Mason, T.G. (2007).
\newblock \enquote{Simulations of complex particle transport in heterogeneous
  active liquids.}
\newblock \emph{Microfluidics and Nanofluidics}, 3(2): 227--237.

\bibitem[{Weiss(2013)}]{weiss13}
Weiss, M. (2013).
\newblock \enquote{Single-particle tracking data reveal anticorrelated
  fractional brownian motion in crowded fluids.}
\newblock \emph{Physical Review E}, 88(1): 010101.

\bibitem[{Weiss et~al.(2004)Weiss, Elsner, Kartberg, and
  Nilsson}]{weiss.etal04}
Weiss, M., Elsner, M., Kartberg, F., and Nilsson, T. (2004).
\newblock \enquote{Anomalous subdiffusion is a measure for cytoplasmic crowding
  in living cells.}
\newblock \emph{Biophysical journal}, 87(5): 3518--3524.

\bibitem[{Wirtz(2009)}]{wirtz09}
Wirtz, D. (2009).
\newblock \enquote{Particle-tracking microrheology of living cells: principles
  and applications.}
\newblock \emph{Annual review of biophysics}, 38: 301--326.

\bibitem[{Wong et~al.(2004)Wong, Gardel, Reichman, Weeks, Valentine, Bausch,
  and Weitz}]{wong.etal04}
Wong, I., Gardel, M., Reichman, D., Weeks, E.R., Valentine, M., Bausch, A., and
  Weitz, D.A. (2004).
\newblock \enquote{Anomalous diffusion probes microstructure dynamics of
  entangled f-actin networks.}
\newblock \emph{Physical review letters}, 92(17): 178101.

\bibitem[{Working et~al.(1997)Working, Newman, Johnson, and
  Cornacoff}]{working.et.al97}
Working, P.K., Newman, M.S., Johnson, J., and Cornacoff, J.B. (1997).
\newblock \enquote{Safety of poly (ethylene glycol) and poly (ethylene glycol)
  derivatives.}
\newblock ACS Publications.

\bibitem[{Zwanzig(2001)}]{zwanzig01}
Zwanzig, R. (2001).
\newblock \emph{Nonequilibrium Statistical Mechanics}.
\newblock New York: Oxford University Press.

\end{thebibliography}

\end{document}